\documentclass[12pt]{amsart}
\usepackage{natbib}
\usepackage{graphicx}
\usepackage{amssymb}
\usepackage{algpseudocode}
\usepackage{tikz}
\usepackage{multirow}
\usepackage{color}
\graphicspath{{figs/}}
\usepackage{setspace}
\usepackage[mathlines]{lineno}
\def\R{\mathbb{R}}
\def\C{\mathcal{C}}

\def\IS{\mathbb{IS}}
\def\IG{\mathbb{IG}}
\def\sgn{\mathop{\hbox{sgn}}}

\def\S{\mathcal S}
\def\F{\mathcal F}
\def\T{\mathcal{T}}

\providecommand{\abs}[1]{\lvert#1\rvert}

\tikzstyle{mybox2} = [fill=yellow!0,
    rectangle, rounded corners, inner sep=5pt, inner ysep=5pt]
\newtheorem{theorem}{Theorem}[section]
\newtheorem{lemma}[theorem]{Lemma}
\newtheorem{proposition}[theorem]{Proposition}
\newtheorem{remark}[theorem]{Remark}
\newtheorem{corollary}[theorem]{Corollary}
\newtheorem{definition}[theorem]{Definition}

\begin{document}
\bibliographystyle{abbrvnat}

\begin{center}
{\Large\sc Coexistence and Extinction in Flow-Kick Systems: \\An invasion growth rate approach}
\end{center}
\vskip 0.25in
\begin{center}
{\sc Sebastian J. Schreiber\\
Department of Evolution and Ecology\\
University of California, Davis USA 95616}
\end{center}
\vskip 0.25in 
 \noindent \noindent \textsc{Abstract.} Natural populations experience a complex interplay of continuous and discrete processes: continuous growth and interactions are punctuated by discrete reproduction events, dispersal, and external disturbances. These dynamics can be modeled by impulsive or flow-kick systems, where continuous flows alternate with instantaneous discrete changes. To study species persistence in these systems, an invasion growth rate theory is developed for flow-kick models with state-dependent timing of kicks and auxiliary variables that can represent stage structure, trait evolution, or environmental forcing. The invasion growth rates correspond to Lyapunov exponents that characterize the average per-capita growth of species when rare.  Two theorems are proven that use invasion growth rates to characterize permanence, a form of robust coexistence where populations remain bounded away from extinction. The first theorem uses Morse decompositions of the extinction set and requires that there exists a species with a positive invasion growth rate for every invariant measure supported on a component of the Morse decomposition. The second theorem uses invasion growth rates to define invasion graphs whose vertices correspond to communities and directed edges to potential invasions. Provided the invasion graph is acyclic, permanence and extinction are fully characterized by the signs of the invasion growth rates. Invasion growth rates are also used to identify the existence of extinction-bound trajectories and attractors that lie on the extinction set.  To demonstrate the framework's utility, these results are applied to three ecological systems: (i) a microbial serial transfer model where state-dependent timing enables coexistence through a storage effect, (ii) a spatially structured consumer-resource model showing intermediate reproductive delays can maximize persistence, and (iii) an empirically parameterized Lotka-Volterra model demonstrating how disturbance can lead to extinction by disrupting facilitation.  Mathematical challenges, particularly for systems with cyclic invasion graphs, and promising biological applications are discussed. These results reveal how the interplay between continuous and discrete dynamics creates ecological outcomes not found in purely continuous or discrete systems, providing a foundation for predicting population persistence and species coexistence in natural communities subject to gradual and sudden changes.\vskip 0.1in
\noindent\textbf{Keywords:} impulsive differential equations, flow-kick systems, uniform persistence, permanence, Lyapunov exponents, growth rates when rare, coexistence, extinction.
 
 
\section{Introduction}
Many biological systems naturally combine continuous and discrete processes. For example, animal populations may grow and interact continuously throughout a season but reproduce only at specific times~\citep{gyllenberg-etal-1997-continuous,pachepsky-etal-2008-reproduction,geng-lutscher-2021-seasonal,geng-etal-2021-coexistence}, diseases may spread continuously while vaccination programs occur in discrete pulses~\citep{agur-etal-1993-vaccination,meng-chen-2008-sir-pulse,donofrio-2002-seir,bai-2015-seirs}, lakes may experience continuous nutrient cycling interrupted by sudden influxes from storm events~\citep{meyer-etal-2018-resilience}, and microbial communities in serial-transfer experiments alternate between continuous growth and dilution events~\citep{yi-dean-2013-bounded,wolfe-dutton-2015-fermented,venkataram-etal-2016-fitness,yurtsev-etal-2016-oscillatory,good-etal-2017-molecular,letten-ludington-2023-pulsed}. Traditional purely continuous or discrete models fail to capture these hybrid dynamics, potentially missing important emergent properties and complex behaviors that arise from the interaction between continuous and discrete processes.

To address this reality, impulsive differential equations (also known as flow-kick models) combine continuous-time and discrete-time processes~\citep{lakshmikantham-etal-1989-impulsive,meyer-etal-2018-resilience}. In these models, continuous processes are modeled by the flow of a system of ordinary differential equations (ODEs). Instantaneous changes in the system (the ``kicks'') occur at specific times and are modeled by an update map. The complete model alternates between the flow and the kicks: the system flows according to its ODEs for a set period of time $\tau$, experiences instantaneous kicks, then flows again, establishing a recurring pattern of disturbance and recovery. The time $\tau$ between kicks can be fixed or state-dependent, depending on the biological context. In models of seasonal reproduction, for example, kicks occur at fixed annual intervals. In contrast, the timing of kicks may depend on system states, as seen in disease models where vaccination pulses are triggered by disease prevalence~\citep{agur-etal-1993-vaccination,donofrio-2002-seir,meng-chen-2008-sir-pulse,bai-2015-seirs}, or in serial transfer experiments where dilution events are determined by the optical density of the medium~\citep{yi-dean-2013-bounded}. The kicks can correspond to changes in population state (e.g. a dilution event in a serial transfer experiment or a reproductive pulse) or changes in flow (e.g. seasonal changes in demographic rates). 

Flow-kick systems present unique mathematical challenges for examining questions of coexistence and extinction. A fruitful approach in theoretical population biology to study these questions uses invasion growth rates—the average per-capita growth rate of a species or genotype when rare~\citep{macarthur-levins-1967-limiting,chesson-1982-environment,chesson-1994-competition,ellner-etal-2016-storage,barabas-etal-2018-chesson,ellner-etal-2019-coexistence,spaak-schreiber-2023-assembly,clark-etal-2024-coexistence}. Mathematically, these invasion growth rates correspond to Lyapunov exponents (a measure of how perturbations grow or shrink over time) and can be used to characterize a strong form of coexistence called \emph{permanence}~\citep{schreiber-2000-permanence,schreiber-etal-2011-persistence,hofbauer-schreiber-2022-permanence}. Permanence occurs when a system maintains a global attractor bounded away from extinction~\citep{schuster-etal-1979-hypercycles,sigmund-schuster-1984-permanence,butler-etal-1986-persistence}, ensuring that coexistence remains robust to both small, repeated perturbations and rare, large disturbances~\citep{jansen-sigmund-1998-permanence,schreiber-2006-persistence}. Invasion growth rates also provide an approach to identify the existence of extinction-bound trajectories -- trajectories where all species are initially present but where a subset of species ultimately tends toward extinction. Although the invasion growth rate approaches to permanence and extinction are well developed for purely continuous or discrete-time models~\citep{rand-etal-1994-evolution,ashwin-etal-1994-attractor,schreiber-2000-permanence,garay-hofbauer-2003-permanence,roth-etal-2017-permanence,patel-schreiber-2018-permanence}, their extension to flow-kick systems remains an open challenge, as it requires careful modifications to handle the interplay between flows and kicks.

This paper addresses these challenges by introducing invasion growth rates for a general class of flow-kick systems that include most forms of autonomous impulsive equations. Using these invasion growth rates, I provide two complementary approaches for establishing permanence: one based on Morse decompositions of the extinction set (collections of invariant sets that capture the fundamental dynamics) and another using invasion graphs that characterize potential transitions between communities. Furthermore, these invasion growth rates are coupled with earlier work on identifying attractors in invariant hyperplanes~\citep{ashwin-etal-1994-attractor,rand-etal-1994-evolution} to identify extinction attractors. To illustrate how to use these results and provide new biological insights into how flow-kick dynamics influences persistence, I analyze models for Lotka-Volterra systems experiencing periodic disturbances, models of serial passage experiments, and spatially structured consumer-resource dynamics with pulsed consumer reproduction. These applications demonstrate specific mechanisms through which the interplay between continuous flows and discrete kicks creates novel persistence outcomes that cannot emerge in purely continuous or purely discrete systems, including kick-induced coexistence in serial transfer experiments and extinction when reproductive events occur too frequently.

The framework developed in this paper provides a foundation for understanding how the complex interplay between continuous and discrete dynamics affects species persistence and community composition in natural systems. By extending invasion growth rate theory to flow-kick models, we gain new insights into predicting when species will coexist or go extinct in environments characterized by both gradual changes and sudden disturbances.

\section{Model Framework}\label{sec:models}

The models developed in this paper track the dynamics of $k$ species (or populations) through their densities $x=(x_1,x_2,\dots,x_k)$ and $\ell$ auxiliary variables $y=(y_1,y_2,\dots,y_\ell)$. This flexible structure allows us to represent various ecological complexities including species interactions, population structure (e.g. spatial, age, or genotypic distributions) and auxiliary processes (e.g. plant-soil feedbacks, seasonal forcing). The species densities take values in the nonnegative cone $[0,\infty)^k$ of the $k$-dimensional Euclidean space $\R^k$, while the auxiliary variables take values in the $\ell$-dimensional Euclidean space $\R^\ell$. Let $z=(x,y)$ correspond to the complete state of the system.

Three components govern the system's impulsive dynamics: a vector field that determines the continuous-time dynamics (``the flow''), a mapping $z\to H(z)$ that governs discrete-time impulses (``the kick''), and a timing function $\tau$ that schedules when discrete events occur. The continuous dynamics are determined by the per capita growth rates $f_i(z)$ of each population $i=1,2,\dots,k$ and the rate of change $g_i(z)$ of the auxiliary variables for $i=1,2,\dots,\ell$. Thus, the equations of motion for the flow are:
\begin{equation}\label{eq:flow}
\begin{aligned}
\frac{dx_i}{dt}=&x_if_i(z)\quad \mbox{ for } i \in \{1,2,\dots,k\}\\
\frac{dy}{dt}=&g(z)\quad \mbox{where } z=(x,y).
\end{aligned}
\end{equation}

To ensure that the flow defined by \eqref{eq:flow} is uniquely determined, I make the following assumption:
\begin{description}
\item[A1] There is a set $\mathcal{K}\subset [0,\infty)^k \times \R^\ell $ open in $[0,\infty)^k \times \R^\ell$ in which the functions $f_i,g_j:\mathcal{K}\to \R$ are locally Lipschitz for $i=1,\dots,k$, $j=1,\dots,\ell$. Consequently, there exists a unique solution $z.t$ to \eqref{eq:flow} for any initial condition $z=(x,y) \in \mathcal{K}$. 
\end{description}

{For each initial condition $z\in K$, the existence and uniqueness theorem implies there exists a maximal open interval $I(z)\subset \R$ such that $0\in I(z)$ and $z.t$ is defined for all $t\in I(z)$~\citep{perko2001differential}.} The mapping $(z,t)\mapsto z.t$ defines the flow of \eqref{eq:flow}, that is, how the initial condition $z$ moves over time following the vector field \eqref{eq:flow}.   The two defining properties of a flow are: (i) $z.0=z$, which means there is no movement if no time has elapsed, and (ii) $(z.t).s=z.(t+s)$ {for $t+s\in I(z)$}, meaning that the system ends up at the same place whether you flow for $t+s$ units of time starting at $z$ or flow for $s$ units of time starting at $z.t$. 

The timing of the kicks is allowed to be state dependent and is given by a continuous positive function $\tau$. To ensure that the model is well defined, I assume that solutions starting with $z$ are defined until at least time $\tau(z)$:

\begin{description}
\item[A2] There is a closed set $\S\subset \mathcal{K}$ and a continuous positive function $\tau: \S \to (0,\infty)$ such that {$\tau(z)\in I(z)$ and $z.t\in\mathcal{K}$} for $0\le t \le \tau(z)$. 
\end{description}

The set $\S$ corresponds to the possible states of the system immediately after a kick. After flowing for $\tau(z)$ units of time, the system enters the set
\begin{equation}\label{eq:states_before_kick}
    \T=\{z.\tau(z): z\in \S\},
\end{equation}
that contains all possible states just before a kick occurs. The system then experiences a kick determined by the mapping:
\begin{equation}\label{eq:kick}
H:\T\to\S \mbox{ where } H(z)=(x_1F_1(z),\dots,x_kF_k(z),G(z))
\end{equation}
where $F_i(z)$ is the multiplicative change in the density of species $i$ and $G(z)$ updates the auxiliary variables. Concerning the kick mapping $H$, the following assumptions are made:
\begin{description}
\item [A3] The functions $F_i:\T\to (0,\infty)$ for $i=1,\dots,k$ are continuous and positive. This ensures that kicks cannot cause immediate extinction of a species that is present.
\item [A4] The mapping $G:\T\to \R^\ell$ is continuous, ensuring continuous updates to the auxiliary variables.
\end{description} 

To construct the full impulsive dynamics, I first consider the state changes between successive kicks given by iterating:
\begin{equation}\label{eq:kick-to-kick-map}
\mbox{the kick-to-kick map }\kappa(z)=H(z.\tau(z)).
\end{equation}

Given an initial condition $z\in \S$, the \emph{forward orbit of $z$} is the sequence 
\[\mathcal{O}^+(z)=\{z,\kappa^1(z),\kappa^2(z),\kappa^3(z),\dots\}\] 
where $\kappa^n(z)=(\underbrace{\kappa\circ \kappa \circ \dots \circ \kappa}_{n\mbox{ times}})(z)$ denotes $n$-fold composition of $\kappa$ with itself. 

To extend the kick-to-kick dynamics to continuous time, I define the mapping $\Phi:[0,\infty)\times \S \to \mathcal{K}$ as follows:
\begin{equation}\label{eq:flow-kick}
\begin{aligned}
\mbox{the flow-kick mapping }\Phi(t,z)&=\left\{ 
\begin{array}{ll}
z & \mbox{for }t=0\\
\kappa^{{n-1}}(z).(t-T(n-1,z)) & \mbox{ for }T(n-1,z)\le t < T(n,z)\\
\kappa^n(z) & \mbox{ for }t=T(n,z)
\end{array}\right.\\
\mbox{where }T(0,z)&=0 \mbox{ and } T(n,z)=T(n-1,z)+\tau(\kappa^{n-1}(z)) \mbox{ for } n=1,2,3,\dots
\end{aligned}
\end{equation}
Here, $T(n,z)$ represents the cumulative time until the $n$-th kick occurs, which is constructed by summing the individual kick intervals $\tau$ for each state in the sequence. This formulation allows us to determine the state of the system at any continuous time $t$, not just at discrete kick times.

To illustrate the diversity of dynamics these models can capture, I present three ecological examples. The first example demonstrates the simplest type of flow-kick system where kicks reduce all populations by a fixed proportion. The second example simultaneously illustrates the use of state-dependent kick times $\tau$ and auxiliary variables $y$ to periodically switch flows. The final example shows how to use auxiliary variables to account for the population structure.

\begin{figure}
\begin{center}
\includegraphics[width=0.45\textwidth]{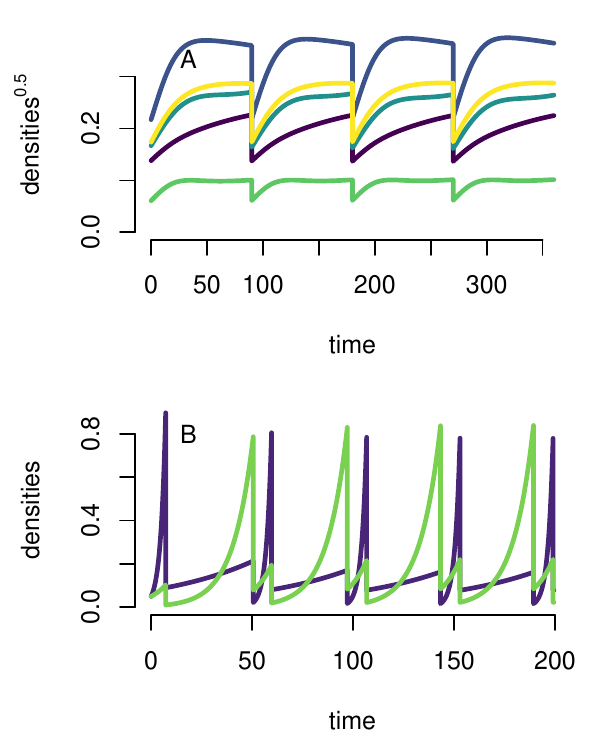}
\includegraphics[width=0.45\textwidth]{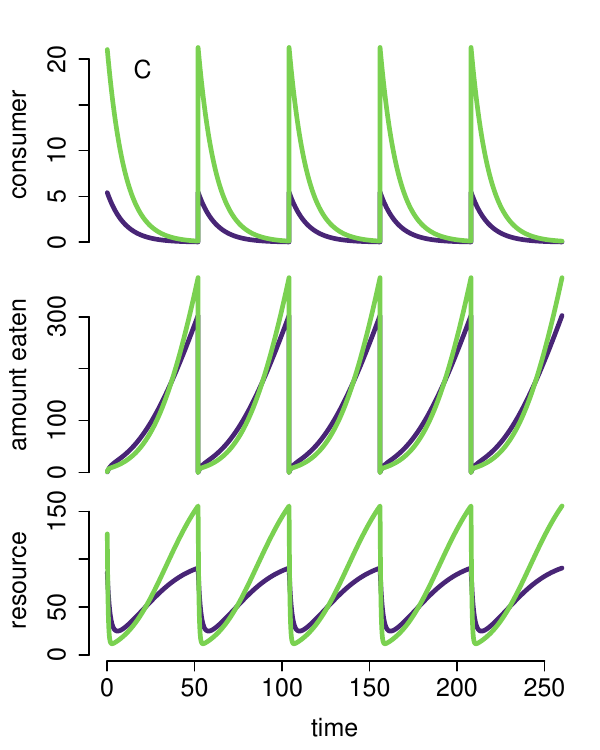}
\end{center}\caption{The flow-kick dynamics of: (A) Lotka-Volterra dynamics with periodic disturbances, (B) competing species experiencing serial passages, and (C) a two-patch model of resource-consumer interactions where the consumer reproduces at discrete time intervals. Parameters: 
\newline
For model (A): $A={\tiny -{\begin{pmatrix}
 0.41 &  0.12 &  0.07 & -0.03 &  0.00 \\
 0.30 &  0.41 &  0.36 &  0.15 &  0.11 \\
 0.07 &  0.36 &  0.31 & -3.60 &  0.34 \\
 0.15 &  0.82 &  0.72 &  0.80 & -0.78 \\
 0.18 &  0.05 &  0.13 &  0.57 &  0.56
\end{pmatrix}}}$, $b={\tiny{\begin{pmatrix}
0.04519 \\
0.10343 \\
0.06709 \\
0.11271 \\
0.07652
\end{pmatrix}}}$, $\tau={90}$, and {disturbance levels $d_i=1-e^{-1}$}. For model (B): $\rho(1)=(0.4,0.1), \rho(2)=(0.02,0.1),\delta=0.1$ and $z(0)=(0.05,0.05,1)$. For model (C): $\alpha=(10,20)$, $\beta=(0.1,0.1)$,$a=(0.1,0.1)$, $m=(0.1,0.1)$,$\theta=0.5$,$\tau=52$,and $d=0.05$}\label{fig:examples}
\end{figure}

\subsection{Lotka-Volterra systems with periodic disturbances}\label{sec:LV}
To analyze how disturbances affect species coexistence, \citet{huston-1979-diversity} studied Lotka-Volterra models with periodic density-independent population reductions. For a community of $k$ species, let $x_i$ denote the density of the $i$-th species, $b_i$ its intrinsic growth rate, and $a_{ij}$ the per-capita effect of species $j$ on the growth rate of species $i$. For this model, $\mathcal{K}=\S=\T=[0,\infty)^k$, meaning the state space for the continuous dynamics and the pre- and post-kick states are all the same. The population dynamics between disturbances follows the classic Lotka-Volterra equations:
\[
\frac{dx_i}{dt} = x_i \left(b_i +\sum_{j=1}^k a_{ij}x_j \right)
\]

Disturbances occur at regular intervals of length $\tau$, reducing the density of each species $i$ by a fixed proportion $d_i\in [0,1)$. These instantaneous reductions can be expressed as the kick map $H: [0,\infty)^k \to [0,\infty)^k $ given by:
\[
H(x_1,\dots,x_k) = ((1-d_1)x_1,\dots,(1-d_k)x_k)
\]

{Figure~\ref{fig:examples}A shows a simulation of this impulsive dynamical system for an empirically parameterized Lotka-Volterra model. The parameterization is based on five grassland species (\emph{Lolium perenne}, \emph{Phleum pratense}, \emph{Plantago.lanceolata}, \emph{Trifolium pratense}, \emph{Trifolium repens}) from field experiments conducted at Teagasc Environmental Research Centre, Co. Wexford, Ireland~\citep{geijzendorffer-etal-2011-sustained}. \citet{geijzendorffer-etal-2011-sustained} estimated the interaction coefficients through field experiments with species grown in monocultures and mixtures of two or three species.  The time scale for the model is in days, and simulations are performed over a growing period of approximately $90$ days. Kicks represent mortality during the non-growing season.}

\subsection{Competition models for serial passage experiments}\label{sec:serial}

\citet{yi-dean-2013-bounded} conducted competition experiments between two populations of \emph{E. coli}: one chloramphenicol resistant (Clm$^R$) and one resistant to tetracycline (Tet$^R$). Fluctuating selection was imposed by alternating sublethal doses of each antibiotic in a minimal glucose medium. When the cultures reached a certain optical density, they were diluted $100$ times in fresh medium with the alternate antibiotic.  To model the dynamics of these competitors, \citet{yi-dean-2013-bounded} used a flow-kick model in which populations grew exponentially until reaching a critical density and then were diluted and introduced to another environment. In this model, $x=(x_1,x_2)$ corresponds to the densities of the two competitors. An auxiliary variable $y\in\{0,1\}$ keeps track of the environments, with chloramphenicol present in environment $0$ and tetracycline present in  environment $1$. Let $\rho_i(y)$ be the per-capita growth rate of genotype $i$ in environment $y$. The flow dynamics represents the uncoupled exponential growth of the two populations:
\[
\begin{aligned}
\frac{dx_i}{dt}=& \rho_i(y)x_i \mbox{ for }i=1,2\\
\frac{dy}{dt}=&0
\end{aligned}
\]

Without loss of generality, we set the critical density for a serial transfer to $1$. Let $\delta<1$ be the factor by which the population is diluted during a serial transfer. Then the post-kick state space $\S$ and the pre-kick state space $\T$ are:
\[
\S=\{x\in [0,\infty)^2: \sum_i x_i = \delta \} \times \{0,1\}
\mbox{ and }
\T=\{x\in [0,\infty)^2: \sum_i x_i = 1\} \times \{0,1\}.
\]

The time to kick $\tau:\S\to (0,\infty)$ is implicitly defined as the unique solution to
\[
\sum_i x_i e^{\rho_i(y)\tau(z)}=1 \mbox{ for }z=(x,y)\in \S.
\]

By the implicit function theorem, $\tau$ is a continuous function. The kick function $H:\T\to\S$ is:
\[
H(z)=(\delta x_1,\delta x_2, (y+1) \mbox{ mod }2) \mbox{ where }z=(x_1,x_2,y)
\] 

This kick function dilutes both populations by the factor $\delta$ and switches the environment between states $0$ and $1$. A sample simulation of this model is shown in Figure~\ref{fig:examples}B.

\subsection{Two Patch Pulsed Consumer-Resource Models}\label{sec:2patch}

\citet{pachepsky-etal-2008-reproduction} introduced a flow-kick model of the dynamics of consumers and resources. In their model, the resource grows continuously and is continuously consumed throughout the year. Consumers accumulate consumed resources and die at a constant rate. Consumer reproduction occurs only once per year as a discrete pulse event proportional to the amount of resources consumed. These models capture the dynamics of consumers, such as barnacles, oysters, or clams, that feed continuously on plankton while breeding once per year.

Here, I consider a two-patch variant of this model to illustrate how auxiliary variables can capture spatial structure and biological feedback. Let $R_i$ denote the density of resources in patch $i=1,2$, $C_i$ the density of the consumer, and $B_i$ the amount of biomass accumulated per consumer. Unlike \citet{pachepsky-etal-2008-reproduction}'s model with logistic growth, resources grow at a constant rate $\alpha_i$ and have a per-capita loss rate $\beta_i$ in patch $i$. This chemostat-type dynamic can represent resources that aren't only locally reproducing but also arriving from regions outside the focal patches, such as plankton arrival near two oyster reef patches. The consumer is sedentary and consumes the resource with a per-capita attack rate $a_i$ in patch $i$. It experiences a per-capita mortality rate $m_i$ in patch $i$. Between reproductive events, the dynamics follow:
\begin{equation}
\begin{aligned}
\frac{dR_i}{dt}=& \alpha_i-\beta_iR_i-a_iR_iC_i\\
\frac{dC_i}{dt}=& -m_i C_i\\
\frac{dB_i}{dt}=& a_i R_i
\end{aligned}
\end{equation}

At fixed intervals of length $\tau>0$, consumers reproduce in proportion to their accumulated biomass, with offspring dispersing between patches with probability $d$. The instantaneous update of population densities follows:
\begin{equation}
\begin{aligned}
&(R_1,R_2,C_1,C_2,B_1,B_2)\mapsto(R_1,R_2,F_1,F_2,0,0)\\
\mbox{where }&F_1=(1-d)C_1(1+\theta  B_1 )+dC_2(1+\theta B_2 )\\
\mbox{and }&F_2=dC_1(1+\theta B_1 )+(1-d)C_2(1+\theta  B_2 )\\
\end{aligned}
\end{equation}
Here, $\theta$ represents the conversion efficiency of the consumed resource to consumers. 

To facilitate analysis of this spatially structured system and align it with our general framework, I express the model using a coordinate transformation that separates the total consumer density from its spatial distribution. This transformation allows one to directly track the size of the consumer population (as our primary state variable), while treating the spatial structure and resource dynamics as auxiliary variables. To express this model in the coordinate system $z=(x,y)$, I define $x=C_1+C_2$ as the total density of consumers, $y_1=R_1$, $y_2=R_2$, $y_3=B_1$, $y_4=B_2$, and $y_5=C_1/x$ as the fraction of consumers in patch $1$. In these coordinates, the flow dynamics become:
\begin{equation}
    \begin{aligned}
    \frac{dx}{dt}&= -m_1y_5x-m_2(1-y_5)x      &\phantom{=} &\frac{dy_5}{dt}= -m_1y_5 + y_5(m_1y_5+m_2(1-y_5))
\\
      \frac{dy_1}{dt}&= \alpha_1-\beta_1 y_1 - a_1 y_1 x y_5 &\phantom{=}
      &\frac{dy_2}{dt}= \alpha_2-\beta_2 y_2 - a_2 y_2 x (1-y_5)\\
      \frac{dy_3}{dt}&= a_1y_1 &\phantom{=}
      &\frac{dy_4}{dt}= a_2 y_2\\
   \end{aligned} 
\end{equation}
and the kick update map is:
\begin{equation}
H(z)=\left(\phi_1(z)+\phi_2(z),y_1,y_2,0,0,\phi_1(z)/(\phi_1(z)+\phi_2(z))\right).
\end{equation}
Here, $\phi_1(z)$ and $\phi_2(z)$ represent the transformed post-reproduction consumer densities $F_1$ and $F_2$ in the original coordinates:
\[
\begin{aligned}
&\phi_1(z)=(1-d)(1+\theta  y_3)xy_5 +d(1+\theta  y_4)x(1-y_5)\\
\mbox{and }&\phi_2(z)=d(1+\theta y_3)xy_5 +(1-d)(1+\theta  y_4)x(1-y_5).
\end{aligned}
\]
The state space is $\mathcal{K}=\S=\T=[0,\infty)\times [0,\infty)^4\times [0,1]$ with $k=1$ (representing the species of the single consumer) and $\ell=5$ (representing the five auxiliary variables). A sample simulation of this model is shown in Figure~\ref{fig:examples}C.

\section{Extinction set and Invasion growth rates}

For flow-kick models, the extinction set $\S_0:=\{z \in \S | \prod_{i=1}^n x_i=0 \}$ is the set of states where at least one species is extinct (i.e., it has a density equal to zero). The per-capita form of the flow and kick mapping imply that for any initial condition $z \in \S_0$, the future states $\kappa^n(z)$ remain in $\S_0$ for all time, capturing the principle of ``no cats, no kittens'' of closed ecological systems. The persistent set $\S_+=\S\setminus \S_0$ corresponds to states in which all species have positive densities. Assumptions \textbf{A1} (via existence and uniqueness) and \textbf{A3} (via positivity) imply that for any initial condition $z\in \S_+$, future states $\kappa^n(z)$ remain in $\S_+$ for all time. Notably, as with nearly all deterministic models, extinction can occur only asymptotically, not in finite time. 

Permanence, defined more precisely in Section~\ref{sec:main}, corresponds to the extinction set $\S_0$ being a repellor, that is, if the system starts with all species present ($z \in \S_+$), then eventually it becomes bounded away from the extinction set.
To understand permanence for flow-kick models, I use invasion growth rates, which measure the average per-capita growth rate of a population when introduced at infinitesimally small densities. {At equilibria of the kick-to-kick map $\kappa$, these invasion growth rates correspond to the logarithms of the appropriate eigenvalues of the Jacobian matrix. To extend the analysis beyond equilibria, it is mathematically convenient to define these invasion growth rates using ergodic theory.} Section~\ref{subsec:IGR} provides some background on ergodic theory, introduces explicit formulas for invasion growth rates, highlights their key properties, and illustrates how to calculate these rates for the previously introduced examples.

\subsection{Invasion growth rates}~\label{subsec:IGR} 
Biologically, invasion growth rates measure how a species would grow if introduced at very low densities into an established community. If the invasion growth rate is positive, the species can increase when rare and potentially establish itself. If negative, the species cannot invade. These rates are fundamental in predicting which species combinations can persist together and which cannot, essentially providing a mathematical foundation for understanding coexistence. To define invasion growth rates, we need some ergodic theory; see, for example, ~\citep{mane-1983-ergodic,katok-etal-1995-dynamical-systems}.  

Consider a forward orbit $\{z,\kappa(z),\kappa^2(z),\dots\}$ of the kick-to-kick map $\kappa:\S\to\S$ (see \eqref{eq:kick-to-kick-map}) for some initial condition $z\in \S$. For any continuous function $h:\S\to \R$ (an ``observable''), it is natural to ask when the long-term average $\lim_{n\to\infty} \frac{1}{n}\sum_{m=0}^{n-1} h(\kappa^m(z))$ is well defined. Ergodic theory provides an answer in terms of invariant measures. A Borel probability measure $\mu$ on $\S$ is \emph{invariant} for the kick-to-kick dynamics \eqref{eq:kick-to-kick-map} if $\int h(z)\mu(dz)=\int h(\kappa(z))\mu(dz)$ for any continuous bounded function $h:\S\to\R$. That is, the average value of $h$ when randomly choosing the initial condition $z$ according to $\mu$ equals the average value of $h$ in the next time step. 

An invariant probability measure $\mu$ is \emph{ergodic} if it cannot be written as a non-trivial convex combination of two invariant probability measures, that is, if $\mu=\alpha \mu_1+(1-\alpha)\mu_2$ for two distinct invariant measures $\mu_1,\mu_2$, then $\alpha=1$ or $\alpha=0$. Intuitively, the ergodic measures represent the most basic statistical behaviors of the system that cannot be further decomposed. The simplest example of an ergodic probability measure is a Dirac measure $\mu=\delta_{z^*}$ associated with an equilibrium $z^*\in \S$ of $\kappa$. This measure is defined by $\int h(z)\mu(dz)=h(z^*)$ for every continuous bounded function $h:\S\to \R$. Similarly, if $z^*$ is a periodic point of period $N$ (i.e., $\kappa^N(z^*)=z^*$), then the probability measure $\mu=\frac{1}{N}\sum_{n=0}^{N-1}\delta_{\kappa^n(z^*)}$ defined by averaging along this periodic orbit is an ergodic measure. More generally, the ergodic theorem~\citep[Chapter 6]{mane-1983-ergodic} implies that for every ergodic measure $\mu$ there exists an initial condition $\tilde z$ such that $\mu$ is determined by averaging along the forward orbit of $\tilde z$, that is, $\int h(z)\mu(dz)=\lim_{N\to\infty} \frac{1}{N}\sum_{n=0}^{N-1} h(\kappa^n(\tilde z))$ for all continuous bounded $h:\S\to\R$. In fact, such points $\tilde z$ make up almost all the support of $\mu$ (i.e., the set of points that do not satisfy this property has a measure zero with respect to $\mu$). This means that for typical initial conditions within the support of an ergodic measure, time averages along trajectories converge to the average with respect to the measure.

Ergodic measures correspond to the long-term statistical behavior of subsets of species, what we might call ecological communities. To see why, for any subset of species $I\subset\{1,2,\dots,k\}$, let
\[
\S_I=\{z\in \S: x_i>0 \mbox{ if and only if } i\in I\}
\]
be the set of states for which only the species in $I$ are present (i.e., have non-zero density). Since these sets $\S_I$ are invariant for kick-to-kick dynamics, an ergodic measure $\mu$ is uniquely associated with a set of species $I(\mu) \subseteq \{1,2,\dots,k\}$, which is the unique subset such that $\mu\left( \S_{I(\mu)} \right)=1$. I call $I(\mu)$ \emph{the species support of the ergodic measure $\mu$}. {For a Dirac measure $\mu=\delta_{z^*}$ at an equilibrium $z^*=(x^*,y^*)$ of the kick-to-kick map $\kappa$, the species support $I(\mu)$ corresponds to the set of species whose densities are positive at the equilibrium, i.e., $i\in \{1,2,\dots,k\}$ such that $x_i^*>0$.}

To define the average per-capita growth rate of species $i$, let $\pi_i:\mathcal{K}\to [0,\infty)$ be the projection of a state $z$ to the density of the $i$-th species, i.e. $\pi_i(z)=x_i$ for $z=(x,y)\in \mathcal{K}$. {As $\int_0^t \frac{c'(t)}{c(t)}dt=\ln \frac{c(t)}{c(0)}$ for any continuously, differentiable, positive function $c:\R\to\R$}, $\ln \frac{\pi_i(\kappa(z))}{\pi_i(z)}$ {for any initial condition $z$ with $\pi_i(z)>0$} equals  
\begin{equation}\label{eq:per-capita_growth_rate}
\mbox{ the per-capita growth rate }r_i(z)=\int_0^{\tau(z)} f_i(z.t)dt+\ln F_i(z.\tau(z)) 
\end{equation}
This represents the per-capita growth rate of species $i$ over a complete flow-kick cycle that begins in the state $z$. Notice that $r_i(z)$ is well defined even if $x_i=0$. In this case, $r_i(z)$ describes the rate at which species $i$ would increase if introduced at infinitesimally small densities. 

For any probability measure $\mu$, we define 
\begin{equation}\label{eq:per-capita_growth_rate_at_mu}
r_i(\mu)=\int r_i(z)\mu(dz) \mbox{, the average per-capita growth rate of species $i$ with respect to $\mu$.}
\end{equation}
$r_i(\mu)$ measures on average how species $i$ would grow if, in the long-term, the state of system follows the distribution given by $\mu$.  If $\mu$ is an ergodic measure, Birkhoff's ergodic theorem implies that 
\begin{equation}\label{eq:BET:one}
\lim_{n\to\infty}\frac{1}{n}\sum_{m=0}^{n-1}r_i(\kappa^m(z))=r_i(\mu)
\end{equation}
for $\mu$-almost every $z$, that is, $\mu\left(\left\{z\in \S\mbox{ such that \eqref{eq:BET:one} holds}\right\}\right)=1$.  {In the case of an ergodic measure $\mu=\frac{1}{N}\sum_{n=1}^N \delta_{\kappa^n(z^*)}$ given by a period-$N$ point $z^*=(x^*,y^*)$, the average per-capita growth rate of species $i$ is given by the arithmetic average of the per-capita growth rates along the periodic orbit, i.e., $r_i(\mu)=\frac{1}{N}\sum_{n=1}^N r_i(\kappa^n(z^*))$. }

{For} an ergodic measure $\mu$, we intuitively expect the average per-capita growth rate of a species supported by $\mu$ to be zero. This is because species densities that remain bounded away from both zero and infinity must have a long-term average growth rate of zero. The following lemma confirms this intuition.

\begin{lemma}\label{lemma:zero}
Let $\mu$ be an ergodic probability measure for the kick-to-kick map $\kappa$ with compact support. Then
$r_i(\mu) = 0$ for all $i \in I(\mu)$. 
\end{lemma}

The proof strategy follows from an argument given for ODE models without auxiliary variables $y$ found in \citep[Lemma 5.1]{schreiber-2000-permanence}. 

\begin{proof}  Since $\mu(\S_{I(\mu)})=1$, Birkhoff's Ergodic Theorem implies that there exists an invariant Borel set $U\subseteq  \S_{I(\mu)}$ such that $\mu(U)=1$ and
\begin{equation}\label{eq:bet}
\lim_{n\to\infty}\frac{1}{n}\sum_{m=0}^{n-1}r_i(\kappa^m(z))=r_i(\mu)
\end{equation}
whenever $z\in U$. Choose an open set $V$ such
that its closure $\overline {V}$ is contained in $\S_{I(\mu)}$, $\overline{V}$ is compact, and $\mu(V\cap
U)>0$. By the Poincar\'{e} recurrence theorem, there exists $z\in
V\cap U$ and an increasing sequence of integers $n_m\uparrow \infty$ such that $\kappa^{n_m}(z)\in V$ for all $m\ge 1$. Since $\overline{V}$ is compact, there exists a
$\delta>0$ such that
    \begin{equation}\label{eq:bounds}
    1/\delta\le \pi_i(\kappa^{n_m}(z)) \le \delta
    \end{equation}
for all $m$. {One has} 
\[
\begin{aligned}
\frac{1}{n}\ln \frac{\pi_i(\kappa^{n}(z))}{\pi_i(z)}=& \frac{1}{n}\ln \prod_{m=0}^{n-1} \frac{\pi_i(\kappa^{m+1}(z))}{\pi_i(\kappa^{m}(z))}=\frac{1}{n}\sum_{m=0}^{n-1}\ln \frac{\pi_i(\kappa^{m+1}(z))}{\pi_i(\kappa^{m}(z))}\\
 =& \frac{1}{n}\sum_{m=0}^{n-1} \left(\int_0^{\tau(\kappa^m(z))} f_i (\kappa^m(z).t)dt+ \ln F_i(\kappa^m(z).\tau(\kappa^m(z))\right)=\frac{1}{n}\sum_{m=0}^{n-1}r_i(\kappa^m(z))
\end{aligned}
\]  
Since $\pi_i(\kappa^{n_m}(z))$ is bounded both above and below by positive constants as shown in \eqref{eq:bounds}, ${\frac{1}{n_m}\ln}\frac{\pi_i(\kappa^{n_m}(z))}{\pi_i(z)}$ must approach zero as $m \to \infty$. Combined with equation \eqref{eq:bet}, this proves that
\[
    r_i(\mu) = \lim_{m\to\infty}\frac{1}{n_m}\sum_{s=0}^{n_m-1}r_i(\kappa^s(z)) =0.
\]
for all species $i \in I(\mu)$.
\end{proof}

Let us revisit the three examples introduced in Section~\ref{sec:models} and derive their invasion growth rates for the ergodic measures supported on the extinction set. 

\subsection{$r_i(\mu)$ for Lotka-Volterra systems with periodic disturbances}\label{sec:LV:ri}

Consider the disturbed Lotka-Volterra models introduced in Section~\ref{sec:LV}. For these systems, given an ergodic measure $\mu(dx)$, the average per-capita growth rate of species $i$ is:
\[
r_i(\mu) = \int\left(\int_0^\tau f_i(x.t)dt\right)\mu(dx) + \ln(1-d_i),
\]
where $f_i(x) = \sum_j a_{ij}x_j + b_i$ represents the per-capita growth rate of species $i$.

These average per-capita growth rates typically admit explicit solutions. To see why, let $x_i^*(\mu)$ denote the average density of species $i$ with respect to $\mu$:
\[
x_i^*(\mu) = \frac{1}{\tau} \int \left(\int_0^\tau \pi_i(z.t) dt\right)\mu(dx)
\]
where $\pi_i:\S\to\R$ is the projection onto the density of species $i$, that is, $\pi_i(z)=x_i$ for $z=(x,y)$.  Ecologically, $x_i^*(\mu)$ represents the time-averaged density of species $i$ in the community represented by $\mu$. These average densities allow us to calculate invasion growth rates without needing to track the detailed temporal dynamics between disturbances. Taking advantage of the linearity of the per-capita growth rate, $r_i(\mu)$ can be expressed in terms of these average densities:
\begin{equation}\label{eq:LV:ri}
r_i(\mu) = \tau\left(\sum_j a_{ij}x_j^*(\mu) + b_i\right) + \ln(1-d_i).
\end{equation}

The average densities $x_i^*(\mu)$ can be determined by observing that: (i) Lemma~\ref{lemma:zero} implies that $r_i(\mu) = 0$ for all species $i \in I(\mu)$ supported by $\mu$, and (ii) $x_i^*(\mu) = 0$ for all species $i \notin I(\mu)$ absent from the support of $\mu$. These properties yield a system of linear equations:
\begin{equation}\label{eq:LV:x*}
\begin{aligned}
x_i^*(\mu) &= 0 & \text{for } i \notin I(\mu) \\
\sum_j a_{ij}x_j^*(\mu) + b_i + \frac{\ln(1-d_i)}{\tau}&= 0 & \text{for } i \in I(\mu)
\end{aligned}
\end{equation}
When this system admits a unique solution (which occurs generically), explicit expressions for all average per-capita growth rates $r_i(\mu)$ can be obtained by substituting the solved values of $x_i^*(\mu)$ into equation \eqref{eq:LV:ri}.

\subsection{$r_i(\mu)$ for Competition models for serial passage experiments}\label{sec:serial:ri}
Consider the competition model in Section~\ref{sec:serial} for serial passage experiments. The kick-to-kick function for this model is:
\[
\kappa(z)=(x_1e^{\rho_1(y)\tau(z)},x_2e^{\rho_2(y)\tau(z)},(y+1) \mbox{ mod }2)\mbox{ where }z=(x_1,x_2,y)
\]
and $\tau$ is implicitly defined by:
\[
\sum_i x_i e^{\rho_i (y) \tau(z)}={1}.
\]

The extinction set $\S_0$ for the kick-to-kick map $\kappa$ consists of four points $(x_1,x_2,y)=(\delta,0,0),(0,\delta,0),(\delta,0,1),(0,\delta,1)$, which correspond to two period-two orbits: $(\delta,0,0)\rightarrow(\delta,0,1)\rightarrow(\delta,0,0)$ and $(0,\delta,0)\rightarrow(0,\delta,1)\rightarrow(0,\delta,0)$. The only ergodic invariant measures in the extinction set are Dirac measures supported on these periodic orbits: $\mu_1(dz)$ with $\mu_1(\{(\delta,0,0)\})=\mu_1(\{(\delta,0,1)\})=1/2$ and $\mu_2(dz)$ with $\mu_2(\{(0,\delta,0)\})=\mu_2(\{(0,\delta,1)\})=1/2$.

For {the} ergodic measure ${\mu_1}$, the average per-capita growth rate of species $i$ is:
\[
r_i(\mu_1)=\int \rho_i(y)\tau(z)\mu_1(dz) + \ln \delta {=\frac{1}{2}\sum_{y=0}^1\rho_i(y)\tau(\delta,0,y).}
\]
{Since }$\tau(\delta,0,y)=\frac{1}{\rho_1(y)}\ln\frac{1}{\delta}$, {one gets} the following expressions for $r_i(\mu_1)$ and, through analogous reasoning, $r_i(\mu_2)$:
\begin{equation}~\label{eq:serial:ri}
r_i(\mu_1) = \frac{\ln \frac{1}{\delta}}{2}\sum_{y \in \{0,1\}} \left(\frac{\rho_i(y)}{\rho_1(y)}-1 \right) \mbox{ and }r_i(\mu_2)= \frac{\ln \frac{1}{\delta}}{2}\sum_{y \in \{0,1\}} \left(\frac{\rho_i(y)}{\rho_2(y)}-1 \right) \mbox{ for }i=1,2.
\end{equation}
{Note: These expressions for $r_i(\mu_j)$ correspond to linearizing $\kappa$ along the period $2$ points and taking the logarithms of the appropriate eigenvalues of the product of the Jacobian matrices.}

{Equation}~\eqref{eq:serial:ri} show how invasion growth rates depend on the ratio of growth rates in each environment, weighted by the time spent in each environment. A species can invade if, on average in both environments, its growth rate relative to the resident species exceeds the penalty for dilution.

\subsection{$r_i(\mu)$ for the Two Patch Pulsed Consumer-Resource Models}\label{sec:2patch:ri}
Consider the two-patch model of consumer-resource interactions in Section~\ref{sec:2patch}. The state space for the kick-to-kick map is $\S=[0,\infty)\times [0,\infty)^4 \times [0,1]$ and the extinction set is $\S_0=\{(0,y)\in \S \}$, which corresponds to the extinction of the consumer.

For an initial condition $z\in \S_0$ on the extinction set, the $y_1,y_2,y_3,y_4$ coordinates of the forward orbit $\{z,\kappa(z),\kappa^2(z),\dots \}$ of the kick-to-kick map converge to $(\alpha_1/\beta_1,\alpha_2/\beta_2,0,0)$. 
At these equilibrated values, the kick-to-kick dynamics of the coordinate $y_5$ corresponds to how the fraction of consumers in patch $1$ changes in the limit of an infinitesimally small consumer population. These dynamics are determined by the fractional linear map:
\[
y_5\mapsto \frac{d\gamma_1  y_5+(1-d)\gamma_2 (1-y_5)}{\gamma_1  y_5+\gamma_2 (1-y_5)}, \text{ where }\gamma_i=e^{-m_i\tau}(1+\theta\tau a_i \alpha_i/\beta_i)
\]
that has a unique globally stable equilibrium:
\begin{equation}\label{eq:y5*}
    y_5^*=\frac{-2\gamma_2+d\gamma_2+d\gamma_1+\sqrt{d^2\gamma_2^2-8d\gamma_1\gamma_2+2d^2\gamma_1\gamma_2+4\gamma_1\gamma_2+d^2\gamma_1^2}}{2\left(\gamma_1-\gamma_2\right)}.
\end{equation}  

Therefore, the only ergodic measure $\mu$ for the kick-to-kick map supported in the extinction set $\S_0$ is the Dirac measure at equilibrium $z^*=(0,\alpha_1/\beta_1,\alpha_2/\beta_2,0,0,y_5^*)$. The average per-capita growth rate of consumers at this ergodic measure is:
\begin{equation}\label{eq:2patch:ri}
   r_i(\mu)=y_5^*\ln \gamma_1 + (1-y_5^*)\ln\gamma_2
\end{equation}

This average per-capita growth rate and the equilibrium value $y_5^*$ relate to the dominant eigenvalue and eigenvector of the extinction equilibrium with respect to the original $(C_1,C_2)$ coordinate system. To see why, the derivative matrix of the kick-to-kick consumer dynamics at the extinction equilibrium $(R_1,R_2,C_1,C_2)=(\alpha_1/\beta_1,\alpha_2/\beta_2,0,0)$ is:
\begin{equation}\label{eq:matrix:2patch}
\begin{pmatrix}
   (1-d)\gamma_1 & d \gamma_2\\
   d \gamma_1 &(1-d)\gamma_2
\end{pmatrix}
\end{equation}
The dominant eigenvalue of this matrix is equal to $e^{r_i(\mu)}$ and the associated eigenvector is $\begin{pmatrix} y_5^* & 1-y_5^*\end{pmatrix}^\intercal$ where $^\intercal$ denotes the transpose. This connection between invasion growth rates and eigenvalues of the linearized system demonstrates how the flow-kick framework extends classical stability analysis. The dominant eigenvalue $e^{r_i(\mu)}$ determines whether small consumer populations grow ($r_i(\mu) > 0$) or decline ($r_i(\mu) < 0$), while the associated eigenvector gives the spatial distribution that would be observed during this growth or decline.

\section{Main results}\label{sec:main}

To characterize permanence using invasion growth rates, I take two complementary approaches. The more general approach uses Morse decompositions in the extinction set $\S_0$, relying on earlier topological characterizations of permanence~\citep{butler-etal-1986-persistence,hofbauer-so-1989-persistence,garay-1989-persistence}. The second approach uses invasion graphs~\citep{hofbauer-schreiber-2022-permanence} that identify possible trajectories that connect different subsets of species according to invasion growth rates. These invasion graphs provide a biologically meaningful way to construct Morse decompositions.

To state the main results, I first assume that the dynamics of the kick-to-kick map are uniformly bounded from above, meaning the system is dissipative. Let $\|z\|=\sum_i |x_i|+\sum_j |y_j|$ denote the $L^1$ norm for $z=(x,y)\in [0,\infty)^k\times \R^\ell$. This norm corresponds to the total density of all species when restricted to the coordinates $x$, that is, $\|x\|=\sum_i x_i$ for $x\in [0,\infty)^k$.

\begin{description}
   \item[A5] The kick-to-kick map $\kappa$ is \emph{dissipative}: There exists $K>0$ such that \[\limsup_{n\to \infty}\|\kappa^n(z)\|\le K\]
   for all $z\in \S$.
\end{description}

The dissipativeness of $\kappa$ implies that there is a global attractor $\Gamma\subset\S$ for $\kappa$. To define this precisely, I recall some definitions from dynamical systems. For a compact set $C\subset \S$, its $\omega$ limit set is defined as:
\[
\omega(C):=\cap_{N\ge 1}\overline{\cup_{n\ge N}\kappa^n(C)}
\]
where  $\kappa^n(C)$ represents the $n$-th iterate of $\kappa$ applied to the set $C$ i.e. $\kappa(C)=\cup_{z\in C}\kappa^n(z)$. The $\omega$-limit set characterizes the long-term behavior of the forward orbits of the kick-to-kick dynamics. A compact set $\Gamma\subset \S$ is an \emph{attractor} if it has a compact neighborhood $U$ such that all trajectories starting in $U$ eventually converge uniformly to $\Gamma$ i.e. $\omega(U)=\Gamma$.  An attractor $\Gamma$ is a \emph{global attractor} if $\omega(z)\subset \Gamma$ for all $z\in \S$.

Permanence ensures that no species goes extinct regardless of initial conditions, provided that all species are initially present. Under Assumption \textbf{A5}, the kick-to-kick map is \emph{permanent} if there is a minimal density $M>0$ such that for any initial condition $z \in \S_+$ with all species present, 
\begin{equation}\label{permanence}
\liminf_{n\rightarrow \infty} \pi_i(\kappa^n(z)) \geq M \mbox{ for all }i=1,2,\dots,k
\end{equation}
where $\pi_i(z)$ corresponds to the density of species $i$. Permanence implies that if all species are initially present, they will coexist in the long term, even when faced with rare but large perturbations or frequent small perturbations \citep{schreiber-2006-persistence}. In terms of attractors, permanence corresponds to a positive attractor $A\subset \S_+$ whose basin of attraction is $\S_+$, which means $\omega(z)\subset A$ for all $z\in \S_+$.

The following lemma proves that permanence of the kick-to-kick map implies the permanence for the continuous-time flow-kick dynamics. A proof is given at the end of Section~\ref{proof:main_theorem}.

\begin{lemma}\label{lem:kick-to-kick-permanence-implies-flow-kick-permanence}
Assume assumptions \textbf{A1}--\textbf{A5} hold. If the kick-to-kick map $\kappa$ is permanent, then the flow-kick dynamics $\Phi$, as defined by \eqref{eq:flow-kick}, is dissipative and permanent: There are constants $\widetilde K\ge \widetilde M>0$ such that
\[
\limsup_{t\rightarrow \infty} \|\Phi(t,z)\| \leq \widetilde K 
\]
whenever $z\in \S$ and
\[
\liminf_{t\rightarrow \infty} \pi_i(\Phi(t,z)) \geq \widetilde M \mbox{ for all }i=1,2,\dots,k
\]
whenever $z\in \S_+$.
\end{lemma}

\subsection{Permanence and Extinction via Morse decompositions} 

I first state conditions for permanence using Morse decompositions, which provide a way to decompose the dynamics on the extinction set into fundamental building blocks. To define a Morse decomposition, I need to introduce the concept of a negative orbit.  A sequence $\mathcal{O}^-=\{z(0),z(-1),z(-2),\dots\} \subset \S$ is a \emph{negative orbit} of $\kappa$ if $\kappa(z(-n))=z(-n+1)$ for all $n=1,2,\dots$. That is, the sequence is consistent with the kick-to-kick dynamics when moving backward in time. As the mapping $\kappa$ need not be a homeomorphism, a negative orbit may not exist for a given initial condition $z\in \S$ or may not be unique when it does exist. However, since the global attractor $\Gamma$ is invariant (i.e., $\kappa(\Gamma)=\Gamma$), negative orbits exist for all initial conditions in $\Gamma$ and can be chosen to lie entirely in $\Gamma$. The $\alpha$-limit set of a negative orbit $\mathcal{O}_-$ is
\[
\alpha(\mathcal{O}_-)=\{\widetilde z \in \S: \lim_{n_k\to \infty} z(-n_k)=\widetilde z \mbox{ for some sequence }n_k\uparrow \infty\}
\]
The $\alpha$-limit set characterizes the asymptotic behavior of the negative orbit in backward time.

Let $\Gamma_0=\Gamma \cap \S_0$ be the restriction of $\Gamma$ to the extinction set and $\Gamma_+=\Gamma\cap \S_+$ be the restriction of $\Gamma$ to where all species persist. {Recall, a} set $A\subset\S$ is \emph{invariant} for $\kappa$ if $\kappa(A)=A$, and an invariant set $A$ is \emph{isolated} if there exists a closed neighborhood $U$ of $A$ such that $A$ is the maximal invariant set in $U$. 

\begin{definition}
A collection of sets $\mathcal{M}=\{M_1, M_2, ..., M_\ell\}$ is a Morse decomposition for $\Gamma_0$ if $M_1, M_2, ..., M_\ell$ are pairwise disjoint isolated invariant compact sets, called Morse sets, such that $M_i\subset \Gamma_0$ and for every $z\in \Gamma_0\backslash \cup_{m=1}^\ell M_m$ there are integers $i< j$ such that $\omega(z)\subset M_j$ and $\alpha(\mathcal{O}_-)\subset M_i$ for all negative orbits $\mathcal{O}^-=\{z(n)\}_{n=-\infty}^\infty \subset \Gamma_0$ with $z(0)=z$. 
\end{definition}

Intuitively, a Morse decomposition divides the extinction set $\Gamma_0$ into isolated invariant pieces that are ordered in such a way that dynamics can only move from lower-indexed to higher-indexed pieces, creating a hierarchical structure of the flow. Morse decompositions always exist but are not necessarily unique. Trivially, one Morse decomposition of $\Gamma_0$ is $\{\Gamma_0\}$ itself. However, Morse decompositions become more useful when they are more refined than this trivial decomposition.

The first main theorem uses Morse decompositions to provide conditions {for permanence} in terms of invasion growth rates. 

\begin{theorem}[Permanence via Morse decompositions]\label{thm:main}
Assume assumptions \textbf{A1}--\textbf{A5} hold. Let $\mathcal{M}=\{M_1, M_2, \dots M_\ell\}$ be a Morse Decomposition for $\Gamma_0$, where $\Gamma_0=\S_0\cap \Gamma$ and $\Gamma$ is the global attractor for the kick-to-kick map $\kappa$. If for each Morse set $M_m$
\begin{equation}\label{eq:thm:main}
\max_ir_i(\mu)>0 \mbox{ for all invariant measures $\mu$ with } \mu(M_m)=1,
\end{equation}
then \eqref{eq:flow-kick} is permanent.
\end{theorem}

Biologically, the condition \eqref{eq:thm:main} means that for each invariant set $M_m$ in the extinction set, there is at least one missing species that can invade. This condition implies that each Morse component $M_m$ is repelling to trajectories with all species initially present. By the minimax theorem, condition \eqref{eq:thm:main} is equivalent to the following condition: there exist positive weights $p_1,\dots,p_k$ (possibly depending on $m$) such that 
\begin{equation}\label{eq:hofbauer:condition}
\sum_{i=1}^k p_i r_i(\mu)>0 \mbox{ for all \emph{ergodic} measures $\mu$ with $\mu(M_m)=1$}
\end{equation}
This equivalent condition is sometimes known as the Hofbauer condition, as it was first introduced by \citet{hofbauer-1981-hypercycles} for permanence of ODEs in the special case of the trivial Morse decomposition (i.e., $\ell =1$ and $M_1=\Gamma_0$). The equivalence of conditions \eqref{eq:thm:main} and \eqref{eq:hofbauer:condition} was first observed by \citet{garay-hofbauer-2003-permanence}.

The proof of Theorem~\ref{thm:main} is based on a topological characterization of permanence due to \citet{hofbauer-so-1989-persistence} and the following lemma. 

\begin{lemma}\label{lem:key:one} For any $z\in \S_+$, there exists an invariant probability measure $\mu$ such that $\mu(\omega(z))=1$ and $r_i(\mu)\le 0$ for all $i\in \{1,2,\dots,k\}$. For any negative orbit $\mathcal{O}^-\subset \Lambda\cap \S_+$, there exists an invariant probability measure $\mu$ such that $\mu(\alpha(\mathcal{O}^-))=1$ and $r_i(\mu)\ge 0$ for all $i\in \{1,2,\dots,k\}$. 
\end{lemma}

Lemma~\ref{lem:key:one} implies that for any initial state with all species present, the average per-capita growth rates associated with its $\omega$-limit set are non-positive and the per-capita growth rates associated with its $\alpha$-limit sets are non-negative. For continuous-time models, this lemma provides a means of simplifying the proof of Theorem~\ref{thm:main} for models without auxiliary variables~\citep{schreiber-2000-permanence} and providing a more direct approach to proving the invasion graph characterization of permanence~\citep{hofbauer-schreiber-2022-permanence}.

I now provide partial converses to Theorem~\ref{thm:main}, giving sufficient conditions for the existence of an extinction-bound attractor and an attractor on the extinction set. Both conditions require that the kick-to-kick map $\kappa$ has additional smoothness.

\begin{proposition}[Extinction bound trajectories]\label{prop:extinction} Assume the kick-to-kick map $\kappa$ is twice continuously differentiable, and assumptions \textbf{A1}--\textbf{A5} hold. If there exists an ergodic measure $\mu$ such that $\mu(\S_0)=1$ and $r_i(\mu)<0$ for all $i\notin I(\mu)$, then there exists $z\in \S_+$ such that $\omega(z)\subset \S_0$. In particular, the system is not permanent. 
 \end{proposition} 
 
This proposition states that if there is a community of species in the extinction set (represented by an ergodic measure $\mu$) that cannot be invaded by any of the missing species (that is, all missing species have negative invasion growth rates), then there exists at least one trajectory starting with all species present that leads to the extinction of some species. This trajectory serves as a counterexample to permanence. The second condition follows from \citep[Theorem 2.12]{ashwin-etal-1994-attractor} and Assumption \textbf{A3}.

 \begin{proposition}[Extinction attractors]\label{prop:extinction:attractor}
     Assume that the kick-to-kick map $\kappa$ is twice continuously differentiable, and assumptions \textbf{A1}--\textbf{A5} hold. If 
     \begin{enumerate}
     \item there is a subset of species $I\subset \{1,2,\dots,k\}$ and an attractor $A\subset \S_I$ for the kick-to-kick $\kappa$ dynamics restricted to $\S_I$, and
     \item $r_i(\mu)<0$ for all $i\notin I$ and for all ergodic measures $\mu$ with $\mu(A)=1$
     \end{enumerate} then $A$ is an attractor for the kick-to-kick map $\kappa:\S\to\S$. 
 \end{proposition}

This proposition provides a systematic way to identify extinction attractors by: (1) verifying the permanence conditions of Theorem~\ref{thm:main} for a subset $I$ of species, and (2) checking that no missing species can invade (the non-invasibility condition $\max_{i\notin I}r_i(\mu)<0$). When both conditions are met, the attractor $A$ persists in the full system.

{\begin{remark} Population models are always approximations of reality. As \citet{box-1979-robustness} observed, ``it would be very remarkable if any system existing in the real world could be exactly represented by any simple model,'', making it crucial to understand whether conclusions drawn from these models remain valid under perturbations of their structural assumptions. In the words of \citet{conley-1978-morse}, ``if such rough equations are to be useful, it is necessary to study them in rough terms.'' In line with this philosophy, \citet{hutson-schmitt-1992-permanence} introduced robust permanence, that is, that permanence holds even with sufficiently small perturbations to the vector field. Using a measure-theoretic approach, \citet{schreiber-2000-permanence} and \citet{roth-etal-2017-permanence} showed that the conditions in Theorem~\ref{thm:main} for permanence also imply robust permanence for purely continuous-time and purely discrete-time models, respectively. Alternatively, using average Lyapunov functions, \citet{garay-hofbauer-2003-permanence} and \citet{patel-schreiber-2018-permanence} provided similar results for continuous-time models, including those with auxiliary variables. Using either of these methods, the permanence criteria presented in Theorem~\ref{thm:main} should extend to robust permanence with respect to perturbations in both the flow and kick growth functions and auxiliary variables, as well as the functional form of the state-dependent kick time $\tau$. Similarly, the conditions in Proposition~\ref{prop:extinction:attractor} for the extinction attractors are also robust to sufficiently small structural perturbations of the governing equations. 
\end{remark}}
\subsection{Permanence and Extinction via Invasion Graphs} 

To help construct the Morse decompositions needed by Theorem~\ref{thm:main}, \citet{hofbauer-schreiber-2022-permanence} introduced invasion graphs as a method to identify Morse decompositions for ODE models. Invasion graphs characterize all possible transitions between communities due to single- or multiple-species invasions. Here, I extend the invasion graph approach to flow-kick systems and show how it can be used to characterize permanence. To define the invasion graph, I need the following additional assumptions:
\begin{description}
   \item[A6] For each ergodic invariant Borel probability measure $\mu$ supported by $\S_0$, $r_j(\mu)\neq 0$ for all $j\notin I(\mu)$, and 
   \item[A7] $\sgn r_j(\mu) = \sgn r_j(\nu)$ for any two ergodic measures $\mu, \nu$ with $I(\mu) = I(\nu)$ and all $j$.
\end{description}

Assumption \textbf{A6} requires that invasion growth rates $r_i(\mu)$ be non-zero for species not supported by $\mu$. This non-degeneracy condition ensures that each species either can or cannot invade a given community. This assumption holds typically for dissipative Lotka-Volterra systems or systems with a finite number of ergodic measures. Assumption \textbf{A7} ensures that all ergodic measures with the same species support agree on which missing species can invade, allowing us to define invasion growth rates for communities rather than for specific ergodic measures. This assumption automatically holds when each face supports at most one invariant probability measure (e.g., there is a unique equilibrium, periodic orbit, or quasi-periodic motion in a given face). Due to their time-averaging property, assumption \textbf{A7} also holds for the periodically disturbed Lotka-Volterra systems introduced in Section~\ref{sec:LV}.

However, \textbf{A7} can fail in some important situations. For example, \citet{mcgehee-armstrong-1977-exclusion} demonstrated coexistence in the sense of a positive attractor but not permanence for a three-species system consisting of two predators and one prey species. In their example, a predator-prey pair coexists along a stable limit cycle at which the other predator can invade ($r_i(\mu)>0$ for the ergodic measure $\mu$ supported by the limit cycle), but the predator-prey pair also supports an unstable equilibrium at which the other predator cannot invade ($r_i(\nu)<0$ for the Dirac measure $\nu$ supported by this equilibrium).

\newcommand{\Com}{\mathcal{C}}
Given assumption \textbf{A7}, I can uniquely define:
\[
r_i(I)   = \sgn r_i(\mu) \mbox{ whenever } I(\mu)=I.
\]
Let \emph{$\Com$ be the set of all communities}: all subsets $I$ of $\{1,2,\dots,k\}$ such that $I = I(\mu)$ for some ergodic measure $\mu$. Let $|\Com|$ be the number of elements in $\Com$. Note that $\Com$ may not include all possible subsets of $\{1,2,\dots,k\}$, but only those that can actually be achieved as the support of some ergodic measure. Following \citet{hofbauer-schreiber-2022-permanence}, \emph{the invasion scheme $\IS$} is the table of signs of invasion growth rates $\{(r_i(I))_{ i\in \{1,2,\dots,k\}}: I\in \Com\}$. This scheme can be represented as a matrix $|\Com| \times k$ where the rows correspond to communities and the columns correspond to species, with entries that indicate whether each species can invade ($+1$), cannot invade ($-1$), or is already present ($0$) in each community. 

Invasion growth rates are useful for identifying possible limit sets $\alpha$ and $\omega$ for an orbit of kick-to-kick dynamics. The following crucial lemma is proved in Section~\ref{proof:IG}. Roughly, the lemma states that if a trajectory starting with a set of species converges to a set supporting only a subset of these species, then the missing species must have negative invasion growth rates in this limit. Conversely, if a negative orbit starting with a set of species converges in backward time to a subset of these species, then the missing species must have positive invasion growth rates in the limit into the deep past. 

\begin{lemma}\label{lemma:critical} Assume \textbf{A1}--\textbf{A7} hold. If $z$ is an initial condition in $\S_I$ for some $I\subset \{1,2,\dots,k\}$ and $\omega(z)$ for the kick-to-kick map $\kappa$ lies in $\S_J$ for some $J\subset I$, then $J\in \Com$ and  $r_j(J)<0$ for all $j\in I\setminus J$. 
Alternatively, if $\mathcal{O}^-=\{z(n)\}_{n=-\infty}^0\subset \Gamma$ is a negative orbit where $z(0)\in \S_I$ for some $I\subset \{1,2,\dots,k\}$ and $\alpha(\mathcal{O}^-)\subset\S_J$ for some $J\subset I$, then $J\in \Com$ and $r_j(J)>0$ for all $j\in I\setminus J$.
\end{lemma} 

Based on this lemma, signs of invasion growth rates can identify possible transitions between communities. The \emph{invasion graph} $\IG$ is the directed graph whose vertices are the communities in $\C$ and for which there is a directed edge from community $I$ to community $J$ if:
\begin{itemize}
    \item $I \neq J$ (communities are distinct),
   \item $r_j(I)>0$ for all $j \in J\setminus I$ (all species in $J$ but not in $I$ can invade community $I$), and 
   \item $r_i(J)<0$ for all  $i \in I \setminus J$ (all species in $I$ but not in $J$ cannot reinvade  community $J$).
\end{itemize}
This directed edge represents a possible transition from community $I$ to community $J$ through ecological invasions.

Importantly, Lemma~\ref{lemma:critical} implies that if a trajectory has its $\alpha$-limit set in $\S_I$ for a proper subset $I\subset \{1,2,\dots,k\}$ and its $\omega$-limit set in $\S_J$ for another proper set $J\neq I$, then there is a directed edge from $I$ to $J$ on the invasion graph. The opposite need not be true: a directed edge from $I$ to $J$ in the invasion graph does not guarantee the existence of an actual trajectory whose $\omega$-limit set lies in $\S_J$ and whose $\alpha$-limit set lies in $\S_I$.

When the invasion graph is acyclic (i.e., contains no directed cycles), the following theorem characterizes permanence. It requires that every community $I\in \C$ with some missing species has at least one missing species with a positive invasion growth rate, that is, $\max_{i\notin I} r_i(I)>0$. This theorem generalizes \citet[Theorem 1]{hofbauer-schreiber-2022-permanence}, which corresponds to the special case of ordinary differential equations.

\begin{theorem}\label{thm:main2}
Assume that \textbf{A1}--\textbf{A7} hold and that $\IG$ is acyclic. If $\max_{i\notin I}r_i(I)>0$ for each $I \in \Com \setminus \{1,2,\dots,k\}$, then \eqref{eq:flow-kick} is permanent.
\end{theorem}

In the development of the theory of coexistence based on invasion growth rates for competitive communities, \citet{chesson-1994-competition} introduced the concept of $-i$ communities: communities that arise when only species $i$ is missing. More precisely, a community $I \subset \{1,2,\dots,k\}$ is called an \emph{$-i$ community} if: (1) $i \notin I$, and (2) $r_j(I)<0$ for all $j\in \{1,2,\dots,k\}\setminus (I\cup\{i\})$. Biologically, an $-i$ community is a subcommunity that cannot be invaded by any species except possibly species $i$. By Lemma~\ref{lemma:critical}, if $z=(x,y)\in \S$ is an initial condition with only species $i$ absent (i.e., $x_i=0$ and $x_j>0$ for $j\neq i$) and $\omega(z)\subset \S_I$ for some $I\subset\{1,2,\dots,k\}$, then $I$ must be an $-i$ community. Theorem~\ref{thm:main2} implies the following corollary:

\begin{corollary} Assume \textbf{A1}--\textbf{A7} hold and $\IG$ is acyclic. If $r_i(I)>0$ for all $-i$ communities with $i\in\{1,2,\dots,k\}$, then \eqref{eq:flow-kick} is permanent. 
\end{corollary}

When the invasion graph is acyclic, Proposition~\ref{prop:extinction} implies that the permanence condition in Theorem~\ref{thm:main2} is sharp. That is, if there is a community $I$ with some missing species and all the missing species have negative per-capita growth rates (i.e., $\max_{i\not\in I}r_i(I)<0$), then the system has an extinction-bound trajectory and is not permanent. Invasion graphs also provide a way to identify extinction attractors:

\begin{corollary}
    [Extinction attractors via Invasion Graphs]\label{cor:extinction} Assume \textbf{A1}--\textbf{A7} hold and the kick-to-kick map is twice continuously differentiable. Let $I\in \Com$ be such that:
\begin{enumerate}
    \item The invasion graph restricted to $I$ is acyclic and satisfies $\max_{j\in I\setminus J} r_j(J)>0$ for any $J\subsetneq I$ with $J\in \Com$, and 
    \item $r_i(I)<0$ for all $i\notin I$,
\end{enumerate}
then there is an attractor $A$ for the kick-to-kick map such that $A\subset \S_I$ and there exists a positive continuous function $\eta:\S_I\to(0,\infty)$ such that $\omega(z)\subset A$ whenever $\pi_i(z)>0$ for $i\in I$ and $\pi_i(z)<\eta(z)$ for $i\notin I$.    
\end{corollary}

The last statement has a clear ecological interpretation: whenever all species in the community $I$ are initially present, any sufficiently small introduction of species outside $I$ will fail to establish, and the system will return to the attractor $A$ where only species in $I$ persist. This provides a mathematical formalization of what ecologists call an ``invasion-resistant'' community.

\section{Applications}

To illustrate how to apply the theoretical results from Section~\ref{sec:main}, I revisit the three examples introduced in Section~\ref{sec:models} to identify the conditions of extinction and permanence. Each example demonstrates different aspects of the theory while raising open research questions and providing new biological insights.

\begin{figure}
\includegraphics[width=0.9\textwidth]{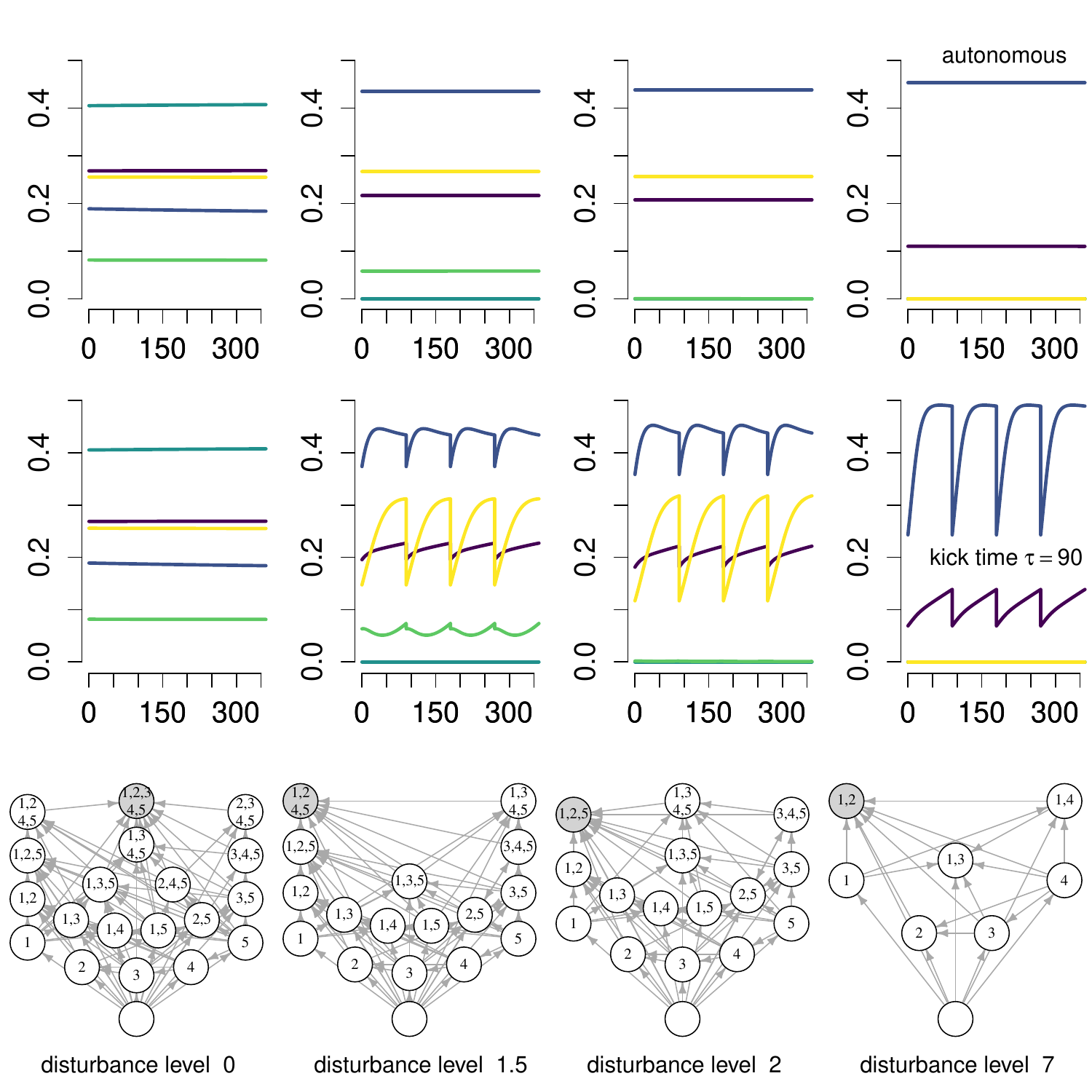}
\caption{The impact of disturbances on an empirically parameterized Lotka-Volterra model. In the top two rows, Lotka-Volterra dynamics of an autonomous model and its flow-kick counterpart are shown at three disturbance levels {$d_i=1-e^{-0.2m}$ for $i=1,2,3,4$ and $d_5=1-e^{-m}$ with $m=0,1.5,2,7$}. In the bottom row, the invasion graph of the autonomous system (which equals the invasion graph of its flow-kick counterpart) is displayed, with nodes representing different species assemblages and arrows indicating possible transitions between communities through invasion processes. Parameter values are as in Figure~\ref{fig:examples}A.}\label{fig:LV}
\end{figure}

\subsection{Invasion graphs for impulsive Lotka-Volterra systems} 
Consider the Lotka-Volterra models with impulsive disturbances introduced in Section~\ref{sec:LV}. For these models, Section~\ref{sec:LV:ri} showed that the invasion growth rate $r_i(\mu)$ is equal to $\left(\sum_j a_{ij} x_j^*(\mu)+b_i\right)\tau+\log(1-d_i)$ where, generically, the mean densities $x_i^*(\mu)$ correspond to the unique solution of a system of linear equations \eqref{eq:LV:x*}.  Due to this uniqueness, assumption \textbf{A6} holds generically. Small perturbations of the model parameters $a_{ij},b_i,d_i$ ensure that the average per-capita growth rates $r_i(\mu)$ are non-zero for missing species $i\notin I(\mu)$, satisfying assumption \textbf{A7}.

The solutions of the linear equations \eqref{eq:LV:x*} always correspond to non-negative equilibria of the autonomous Lotka-Volterra equations:
\begin{equation}\label{eq:LV:autonomous}
\frac{dx_i}{dt}=x_i\left(\sum_j \tau a_{ij}x_j + \tau b_i+ \ln(1-d_i) \right)
\end{equation}
where the $\ln(1-d_i)$ terms account for the average effects of the kicks on the long-term dynamics. Moreover, the invasion growth rates of the flow-kick system correspond to the invasion growth rates of this autonomous system. 

Due to this mathematical relationship, the invasion graph $\mathcal{G}_{\mbox{\tiny flow-kick}}$ for the flow-kick system is a subgraph of the invasion graph $\mathcal{G}_{\mbox{\tiny auto}}$ for the autonomous system \eqref{eq:LV:autonomous}. However, $\mathcal{G}_{\mbox{\tiny flow-kick}}$ may lack some vertices of $\mathcal{G}_{\mbox{\tiny auto}}$ and their corresponding incoming and outgoing edges. Whether this occurs is an open question with ecological implications: if the graphs differ, it would suggest that disturbances fundamentally alter which community states are possible, beyond merely changing their stability properties. If $\mathcal{G}_{\mbox{\tiny auto}}$ is acyclic and satisfies the invasion condition, then $\mathcal{G}_{\mbox{\tiny flow-kick}}$ is acyclic and satisfies the invasion condition, and consequently, the flow-kick system is permanent.

An important class of shared vertices between $\mathcal{G}_{\mbox{\tiny auto}}$ and $\mathcal{G}_{\mbox{\tiny flow-kick}}$ corresponds to subsystems of the autonomous model that meet the permanence conditions of Theorem~\ref{thm:main2}. Specifically, let $I\subset \{1,2,\dots,k\}$ correspond to a vertex of the autonomous invasion graph $\mathcal{G}_{\mbox{\tiny auto}}$ and define $\mathcal{G}_{\mbox{\tiny auto},I}$ to be the subgraph of $\mathcal{G}_{\mbox{\tiny auto}}$ with all vertices $J\subset \{1,2,\dots,k\}$ such that $J\subseteq I$ and all the corresponding edges. If $\mathcal{G}_{\mbox{\tiny auto},I}$ is acyclic and satisfies the invasion criterion of Theorem~\ref{thm:main2}, then the same holds for the subgraph $\mathcal{G}_{\mbox{\tiny flow-kick},I}$ of $\mathcal{G}_{\mbox{\tiny flow-kick}}$. Therefore, the subsystem of the flow-kick system corresponding to the species in $I$ is permanent. By the {\citet{brouwer-1911-uber}} fixed point theorem (which guarantees the existence of fixed points {for continuous maps of a convex set into itelf)}, there exists a fixed point $z^*\in \S_I$ on the kick-to-kick map and, consequently, an ergodic measure $\mu=\delta_{z^*}$ with $I(\mu)=I$. Therefore, $I$ is a vertex of $\mathcal{G}_{\mbox{\tiny flow-kick}}$.

To illustrate how to utilize these results, consider {\citet{geijzendorffer-etal-2011-sustained}'s empirically parameterized Lotka-Volterra model described in Section~\ref{sec:LV}.} In the absence of disturbances, the invasion graph for this community is acyclic and satisfies the permanence condition: all five species coexist.

Now consider disturbances that affect species $5$ {(\emph{Trifolium repens}) more strongly than the other species} (Figure~\ref{fig:LV}). For each disturbance level, I simulate the dynamics of the flow-kick system for $\tau={90}$ (middle row of Figure~\ref{fig:LV}) and the corresponding autonomous system (top row of Figure~\ref{fig:LV}). For the autonomous system, I also calculated its invasion graph $\mathcal{G}_{\mbox{\tiny auto}}$ at all disturbance levels (bottom row of Figure~\ref{fig:LV}) using the algorithm introduced in \citet{hofbauer-schreiber-2022-permanence}. For each disturbance level, these invasion graphs are acyclic, and the subgraph corresponding to every vertex is acyclic and satisfies the permanence condition. Hence, the invasion graph for the autonomous system is also the invasion graph for the flow-kick system: $\mathcal{G}_{\mbox{\tiny auto}}=\mathcal{G}_{\mbox{\tiny flow-kick}}$.

In the absence of disturbances (${m=0}$), the autonomous and flow-kick systems are identical and the entire system is permanent (first column of Figure~\ref{fig:LV}). Simulations suggest that coexistence occurs at a globally stable equilibrium. At intermediate disturbance levels (${m=1.5}$), the invasion graphs $\mathcal{G}_{\mbox{\tiny auto}}=\mathcal{G}_{\mbox{\tiny flow-kick}}$, in conjunction with Proposition~\ref{cor:extinction}, predict an attractor that excludes species {$3$} (second column of Figure~\ref{fig:LV}). At higher disturbances (${m=2}$), the invasion graphs predict an attractor excluding species {$3$} and {$4$} (third column of Figure~\ref{fig:LV}). At sufficiently high levels of disturbance-induced mortality rates (${m=7}$), the invasion graphs predict an attractor that excludes species {3},{4}, and {5}.

Interestingly, species $5$ is not the first to be excluded despite being the  species {most impacted} by disturbance-induced mortality. This unexpected result arises from the complex network of interactions: species 5 exerts positive direct and indirect effects on species {$3$} and {$4$}. These interaction effects buffer species 5 against disturbance while simultaneously making its competitors more vulnerable to extinction. This highlights how indirect effects in ecological networks can lead to counterintuitive outcomes in response to targeted perturbations.

\begin{figure}
    \centering
    \includegraphics[width=0.4\textwidth]{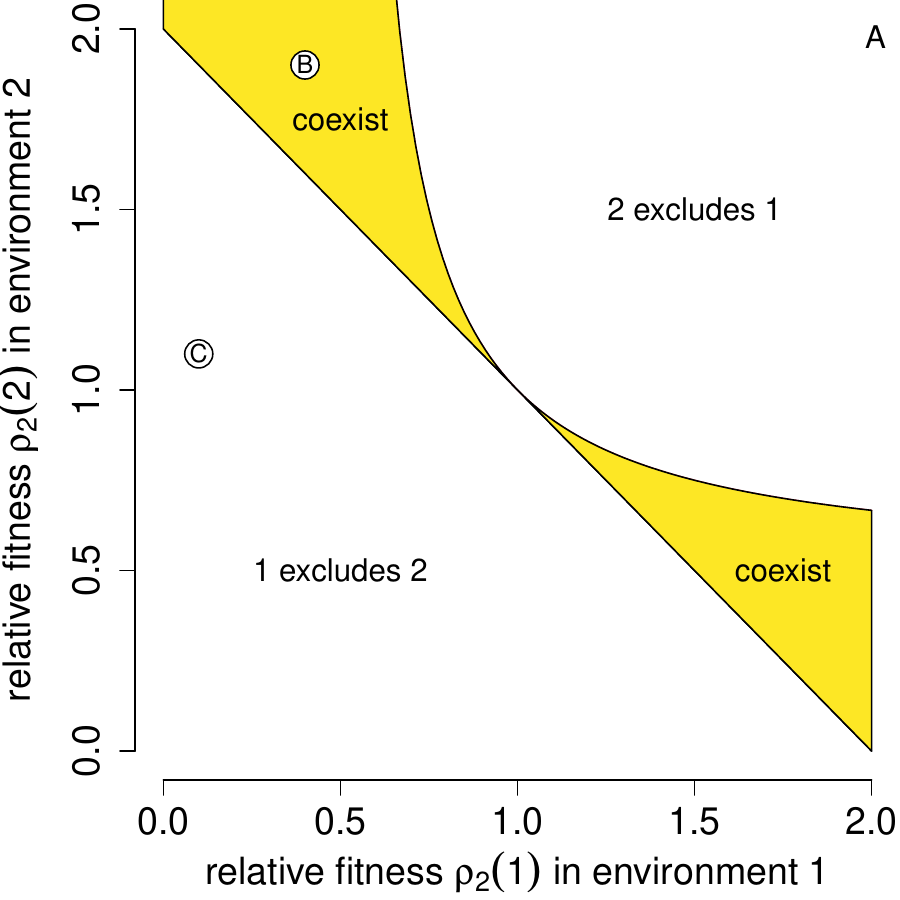}
    \includegraphics[width=0.275\textwidth]{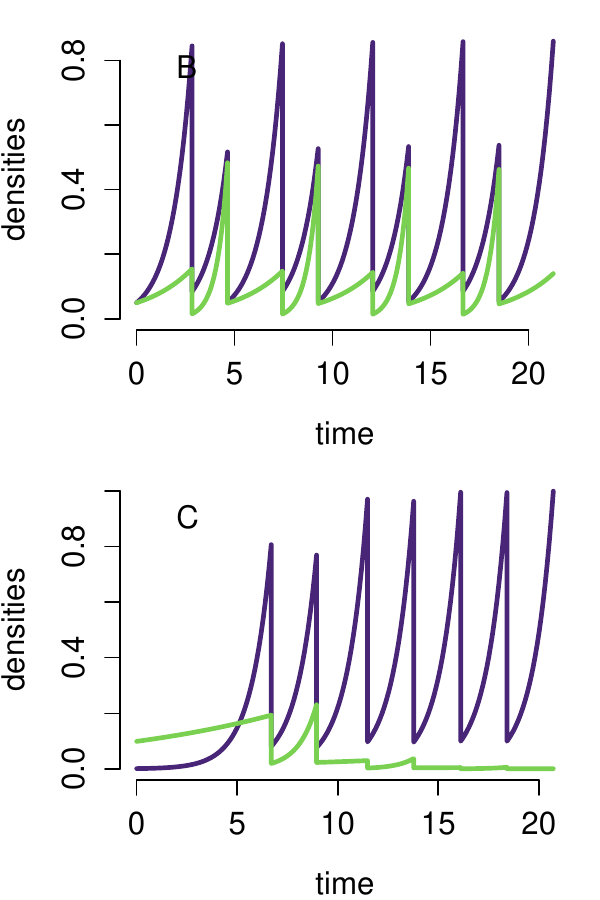}
    \caption{Coexistence and exclusion in periodic serial passage experiments. (A) Regions of coexistence and exclusion based on the signs of the average per-capita growth rates. (B) Competitive dynamics for the parameter combinations indicated in (A). Parameters: $\delta=0.1$.}
    \label{fig:YiDean}
\end{figure} 

\subsection{Coexistence and exclusion of competitors in serial passage experiments}\label{sec:serial:applications}

Recall the competition model in Section~\ref{sec:serial} with extinction set $\S_0=\{(\delta,0,0),(0,\delta,0),(\delta,0,1),(0,\delta,1)\}$. For a Morse decomposition of $\S_0$, I divide it into two sets $M_1=\{(\delta,0,0),(\delta,0,1)\}$ and $M_2=\{(0,\delta,0),(0,\delta,1)\}$ consisting of the two periodic orbits of a single species of the kick-to-kick map $\kappa$. 

Section~\ref{sec:serial:ri} identified the unique ergodic measures $\mu_1,\mu_2$ supported on the sets $M_1$ and $M_2$, respectively. The average per-capita growth rates $r_i(\mu_j)$ at these ergodic measures are given by \eqref{eq:serial:ri}. Lemma~\ref{lemma:zero} implies $r_1(\mu_1)=r_2(\mu_2)=0$. By Theorem~\ref{thm:main}, the model is permanent if:
\[
\frac{1}{2}\sum_{y=0,1} \frac{\rho_2(y)}{\rho_1(y)}>1  \mbox{ and }\frac{1}{2}\sum_{y=0,1} \frac{\rho_1(y)}{\rho_2(y)}>1.
\]
Equivalently, 
\begin{equation}\label{eq:serial:permanence}
\frac{1}{2}\sum_{y=0,1} \frac{\rho_1(y)}{\rho_2(y)}>1 > \frac{1}{\frac{1}{2}\sum_{y=0,1} \frac{\rho_2(y)}{\rho_1(y)}}.
\end{equation}
Condition~\eqref{eq:serial:permanence} is equivalent to the condition found in \citet{yi-dean-2013-bounded}. The left-hand side of \eqref{eq:serial:permanence} corresponds to the arithmetic mean of the ratios $\frac{\rho_1(0)}{\rho_2(0)}$ and $\frac{\rho_1(1)}{\rho_2(1)}$, while the right-hand side corresponds to the harmonic mean of these ratios. Thus, coexistence requires that the arithmetic mean of the relative fitnesses for each species is greater than the corresponding harmonic mean. This creates a ``fitness balancing'' requirement for coexistence.

When one of the inequalities in \eqref{eq:serial:permanence} is violated, Proposition~\ref{prop:extinction} implies that one of the periodic orbits in the extinction set is an attractor and, consequently, the system is not permanent. For example, if 
\begin{equation}\label{eq:serial:one_goes_extinct}
\frac{1}{2}\sum_{y=0,1} \frac{\rho_1(y)}{\rho_2(y)}<1
\end{equation}
then $r_1(\mu_2)<0$ and $M_2$ is an attractor. Moreover, according to Jensen's inequality, if \eqref{eq:serial:one_goes_extinct} holds, then 
\[ \frac{1}{\frac{1}{2}\sum_{y=0,1} \frac{\rho_2(y)}{\rho_1(y)}}\le {\frac{1}{2}\sum_{y=0,1} \frac{\rho_1(y)}{\rho_2(y)}}<1. \] 
Thus, $r_2(\mu_1)>0$ and $M_1$ is a repeller. This suggests that species $2$ excludes species $1$. However, whether exclusion occurs for all positive initial conditions is an open question, as the attractor-repeller structure established by our theoretical results guarantees exclusion only for some set of initial conditions but does not necessarily determine the full basin of attraction for the exclusion attractor. Using a similar argument for the reverse of the second inequality in \eqref{eq:serial:permanence} implies that, generically, the species coexist or one excludes the other. Notably, there is no bi-stability in the system. Figure~\ref{fig:YiDean}A illustrates the main conclusions of this analysis.

\begin{figure}[t!!]
    \centering
    \includegraphics[width=0.4\textwidth]{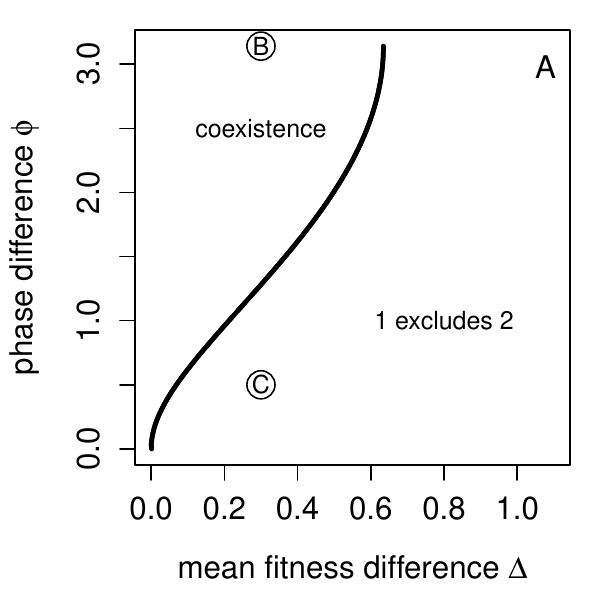}
    \includegraphics[width=0.275\textwidth]{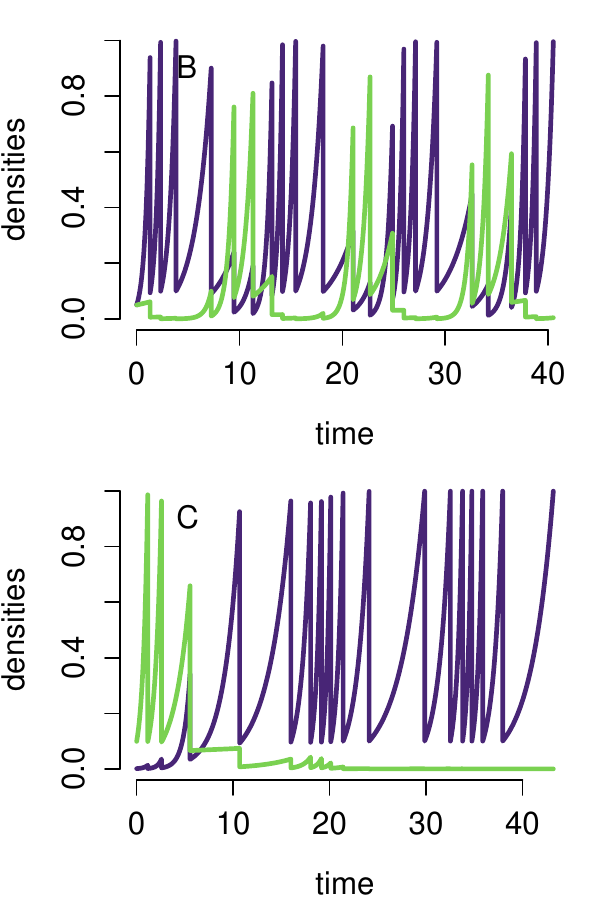}
    \caption{Coexistence and exclusion in quasi-periodically forced serial passage experiments. (A) Regions of coexistence and exclusion based on the signs of the average per-capita growth rates. (B) Competitive dynamics for the parameter combinations indicated in (A). Parameters: $\rho_1(\cos\theta,\sin\theta)=1.1+\Delta+\sin(\theta)$ where $\Delta$ is the mean fitness difference, $\rho_2(\cos\theta,\sin\theta)=1.1+\sin(\theta+\phi)$ and $\psi=1$. See the end of Section~\ref{sec:serial:applications} for the modified model description.}
    \label{fig:YiDean-Quasi}
\end{figure} 

All of these arguments extend to more complex environmental dynamics as characterized by the auxiliary variable $y$. For example, one can replace the two-periodic kick dynamics in $\{0,1\}$ with a quasi-periodic motion on the circle. Specifically, let $S^1=\{(\cos \theta,\sin \theta):\theta\in \R\}$ be the unit circle in $\R^2$. Let $\rho_1,\rho_2:S^1\to (0,\infty)$ be two continuous positive functions. As before, define the flow by $\frac{dx_i}{dt}=\rho_i(y)x_i, \frac{dy}{dt}=0$. The time to kick $\tau(x,y)$ is implicitly given by $\sum_i x_i \exp(\rho_i(y)\tau)=1/\delta$. The kick map is $H(x_1,x_2,\cos \theta,\sin\theta)=(\delta x_1, \delta x_2,\cos(\theta+\psi), \sin(\theta+\psi))$, where $\psi/\pi$ is irrational. In this case, the average per-capita growth rates in the extinction set are:
\[
\left(\ln \frac{1}{\delta}\right)\int_0^{2\pi} \left(\frac{\rho_i(\cos\theta,\sin\theta)}{\rho_j(\cos\theta,\sin\theta)}-1 \right) \frac{d\theta}{2\pi} \mbox{ with }i=1,2 \mbox{ and }j=1,2 
\]

For the quasi-periodic case illustrated in Figure~\ref{fig:YiDean-Quasi}, we set $\rho_1(\cos\theta,\sin\theta)=1.1+\Delta+\sin(\theta)$ where $\Delta$ represents the mean fitness difference, and $\rho_2(\cos\theta,\sin\theta)=1.1+\sin(\theta+\phi)$ with $\psi=1$ determining the rotation angle in the kick map. {This figure highlights that the phase difference $\phi$, a measure of nice differentiation, needs to be sufficiently large relative the mean fitness difference $\Delta$ to ensure coexistence.}

\subsection{Persistence and Extinction of the Two-Patch Pulsed Consumer-Resource Model}

\begin{figure}
    \centering
    \includegraphics[width=0.4\textwidth]{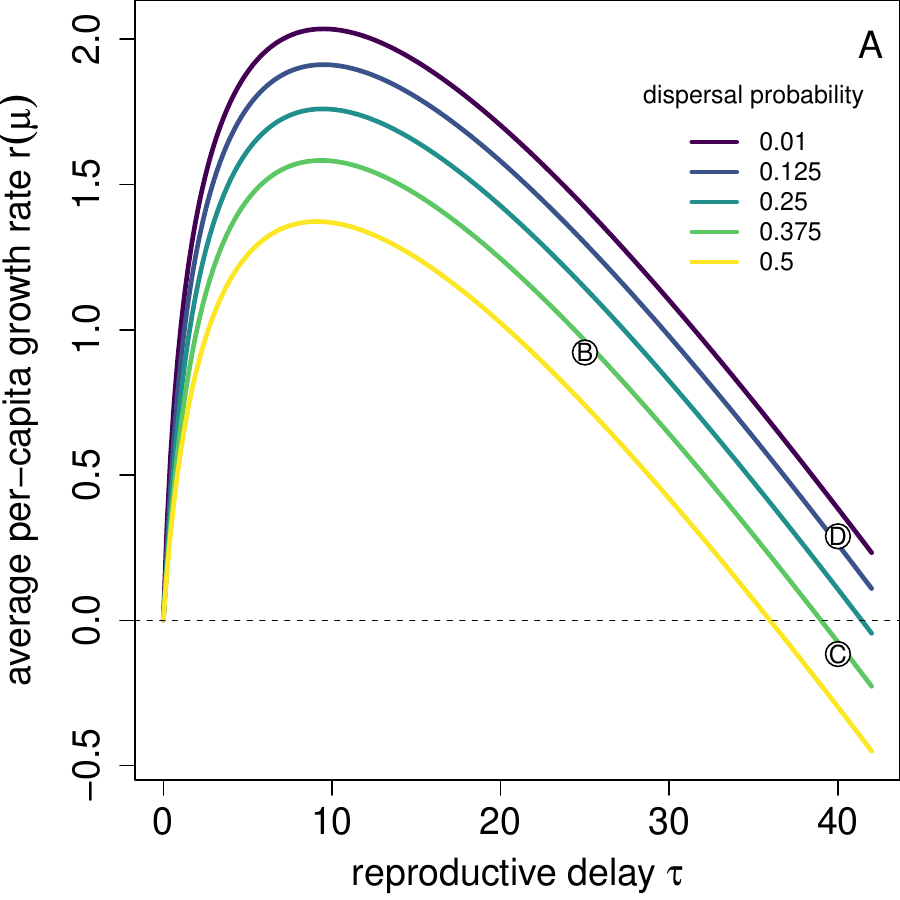}
    \includegraphics[width=0.4\textwidth]{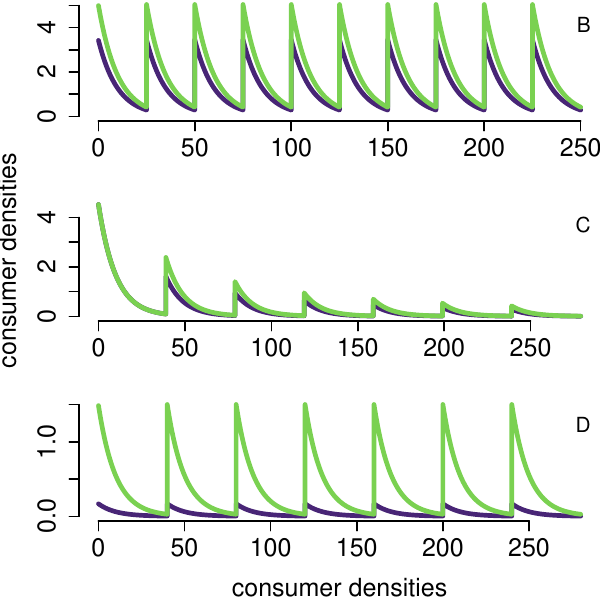}
    \caption{Impacts of reproductive delays and natal dispersal on consumer persistence. (A) The average per-capita growth rate $r(\mu)$ at the extinction set as a function of the reproductive delay $\tau$ for different dispersal probabilities. (B)-(D) The consumer dynamics for the reproductive lags $\tau$ and dispersal probabilities $d$ indicated in (A). Parameters: $\alpha=(0,4)$, $\beta=(0.1,0.1)$, $a=(0.1,0.1)$, $m=(0.1,0.1)$, $\theta=0.5$, $\tau$ as shown, and $d=0.0, 0.125, 0.25, 0.375, 0.5$.}
    \label{fig:consumer-resource}
\end{figure} 

Consider the model introduced in Section~\ref{sec:2patch} with a consumer living in two patches. For this model, Section~\ref{sec:2patch:ri} showed that there was a unique ergodic measure $\mu$ on the extinction set $\S_0$. This invariant measure corresponds to the consumer-free equilibrium $z^*=(0,\alpha_1/\beta_1,\alpha_2/\beta_2,0,0,y_5^*)$ where $y_5^*$, given by \eqref{eq:y5*}, corresponds to the stable patch distribution of an infinitesimally small consumer population at the consumer-free equilibrium. The average per-capita growth rate $r(\mu)$ of the consumer at this equilibrium is given by \eqref{eq:2patch:ri}.

By Theorem~\ref{thm:main} with the trivial Morse decomposition $M_1=\{z^*\}$ of $\S_0$, $r(\mu)>0$ implies that the consumer persists in the sense of permanence. Alternatively, if $r(\mu)<0$, then Proposition~\ref{prop:extinction} implies that the extinction set is an attractor. Furthermore, since the resource density in patch $i$ is always $\leq \alpha_i/\beta_i$ on the global attractor $\Gamma$, $r(z)\leq r(\mu)<0$ for any $z\in \Gamma$. Consequently, the consumer-free equilibrium is globally stable whenever $r(\mu)<0$.

To explore some biological implications of this persistence condition, I consider how $r(\mu)$ depends on the dispersal probability $d$ and the reproductive delay $\tau$. As $r(\mu)$ corresponds to the dominant eigenvalue of the matrix \eqref{eq:matrix:2patch}, the reduction principle~\citep{feldman-liberman-1986-modifiers,kirkland-etal-2006-dispersal,altenberg-etal-2017-reduction}--which states that in heterogeneous environments, random dispersal reduces population growth rates--implies that $r(\mu)$ is a decreasing function of the dispersal rate $d$ whenever $\gamma_1\neq \gamma_2$ where $\gamma_i = e^{-m_i\tau}(1+\theta\tau a_i \alpha_i/\beta_i)$, that is, whenever there is spatial heterogeneity. In the limits of no dispersal and complete dispersal, we have:
\[\lim_{d\downarrow 0}r(\mu)=\max_i \ln \gamma_i \mbox{ and }\lim_{d\uparrow 1}r(\mu)=\ln \frac{\gamma_1+\gamma_2}{2} . \] 
Therefore, whenever $\frac{\gamma_1+\gamma_2}{2}<1< \max_i \gamma_i$, there is a critical dispersal rate $d^*$ below which the consumer persists and above which the consumer goes extinct.

Alternatively, to understand the impact of reproductive delay $\tau$ on consumer persistence, I explore the behavior of $r(\mu)$ near $\tau=0$ and in the limit $\tau\to\infty$. For $\tau=0$, the matrix \eqref{eq:matrix:2patch} is the doubly stochastic matrix $\begin{pmatrix}
    1-d&d\\ d& 1-d
\end{pmatrix}$ with dominant eigenvalue $1$, dominant right eigenvector $v=\begin{pmatrix}
    \frac{1}{2}\\\frac{1}{2}
\end{pmatrix}$, and dominant left eigenvector $w=\begin{pmatrix}
    1& 1
\end{pmatrix}$. Let $\lambda(\tau)$ be the leading eigenvalue of the matrix $A(\tau)$ given by \eqref{eq:matrix:2patch}. \citet[Theorem 6.3.12]{horn-johnson-1994-matrix} implies that:
\[
\lambda'(0)=wA'(0)v=\frac{1}{2}\sum_i(\theta a_i \alpha_i/\beta_i - m_i). 
\]
In particular, $\lambda(\tau)$ increases for small reproductive delays if $\frac{1}{2}\sum_i(\theta a_i \alpha_i/\beta_i - m_i)>0$. Alternatively, $\lim_{\tau\to\infty}\lambda(\tau)=0$. Hence, if $\frac{1}{2}\sum_i \theta a_i \alpha_i/\beta_i> \frac{1}{2}\sum_i m_i$ then there are two critical values of $\tau$, $0<\tau_*\le \tau^*$, such that $r(\mu)>0$ for $0<\tau<\tau_*$ and $r(\mu)<0$ for $\tau>\tau^*$. I conjecture that $\tau_*=\tau^*$.

Figure~\ref{fig:consumer-resource} illustrates the simultaneous effects of reproductive delay $\tau$ and dispersal probability $d$ on consumer persistence. In this figure, patch $1$ has no resource production ($\alpha_1=0$) and therefore does not support consumer reproduction, while patch $2$ has positive resource production ($\alpha_2=4$) sufficient to support consumer reproduction. As $\frac{1}{2}\sum_i \theta a_i \alpha_i/\beta_i>\frac{1}{2}\sum_i m_i$, $r(\mu)$ is positive for sufficiently small reproductive delays and negative for sufficiently long reproductive delays (Figure~\ref{fig:consumer-resource}A). Consistent with my conjecture that $\tau_*=\tau^*$, the low-density consumer growth rate $r(\mu)$ has a unique maximum with respect to $\tau$. Moreover, consistent with the reduction principle, $r(\mu)$ decreases with the probability of dispersal. Hence, intermediate reproductive delays or lower dispersal probabilities are required for consumer persistence (Figure~\ref{fig:consumer-resource}B--D).

\section{Discussion} 

Natural populations experience a complex interplay of continuous and discrete processes: continuous growth and interactions are punctuated by discrete reproduction events, dispersal, and external disturbances. This work developed a mathematical framework that characterizes permanence in flow-kick systems through invasion growth rates. By accommodating auxiliary variables and the state-dependent timing of discrete events, the framework extends beyond previous work on purely continuous or discrete systems~\citep{roth-etal-2017-permanence,patel-schreiber-2018-permanence,hofbauer-schreiber-2022-permanence}. However, several mathematical challenges and opportunities for improvement remain. Applying these results to three ecological models reveals how the interplay between continuous and discrete dynamics can lead to some counterintuitive predictions for population persistence and ecological outcomes.

\subsection{Mathematical Advances and Challenges}

The permanence results in this paper involved defining the invasion growth rates and proving the permanence of the discrete-time kick-to-kick map $\kappa$. The assumptions of dissipativeness, positivity and continuity allowed the extension of permanence from the discrete-time map to the continuous-time flow-kick system (Lemma~\ref{lem:kick-to-kick-permanence-implies-flow-kick-permanence}).

Theorem~\ref{thm:main} extends the work of \citet{roth-etal-2017-permanence} of discrete-time maps in three important ways. First, it generalizes their framework through the use of auxiliary variables. This generalization is significant because, although \citet{roth-etal-2017-permanence} developed invasion growth rate characterizations of permanence for discrete-time maps with stage structure, our approach encompasses a much broader class of systems. As illustrated by the spatially structured consumer model in Section~\ref{sec:2patch}, stage structure can be incorporated by a change of variables corresponding to the total population density (i.e., $x_i$) and the frequencies in the different stages (that is, part of the variables in the auxiliary state $y$). However, our auxiliary variables can also represent deterministic forcing (e.g. periodic, quasiperiodic, chaotic) of the community dynamics as illustrated in Section~\ref{sec:serial:applications}, trait evolution \citep{patel-schreiber-2018-permanence}, plant-soil feedbacks, and many other processes. Moreover, by using auxiliary variables, our theorem makes no assumptions about the irreducibility of stage structure models, in contrast to \citet{roth-etal-2017-permanence}. Building on this first advancement, the second contribution is methodological: our proof of Theorem~\ref{thm:main} is substantially simpler than the approach in \citet{roth-etal-2017-permanence}. Their method required extending the models to skew product systems in an infinite-dimensional sequence space (i.e., the inverse limit of the system) and relied on non-autonomous Perron-Frobenius theorems, such as those in \citet{ruelle-1979-exponents}. Our approach avoids these complications. Finally, from a theoretical perspective, Theorem~\ref{thm:main2} provides the first proof that uses invasion graphs to characterize permanence for discrete-time models.

The key lemma for the proofs (Lemma~\ref{lem:key:one}) simplifies the earlier proofs using ergodic theory to characterize permanence~\citep{schreiber-2000-permanence,hofbauer-schreiber-2010-robust,roth-etal-2017-permanence,hofbauer-schreiber-2022-permanence} and provides a unified approach to the two permanence theorems and the different model types. In forward time, this lemma has a simple biological interpretation. If all species are initially present, then the invasion growth rates associated with the community trajectory are non-positive as their densities remain bounded. Conversely, going backward in time, (bounded) trajectories are associated with non-negative invasion growth rates. Namely, they were either rare in the past and became more common (i.e., a positive invasion growth rate)  or they were common and remained common (i.e. a zero invasion growth rate). 

Despite this progress, there are multiple mathematical opportunities and open challenges. Most importantly, in the case of cyclic invasion graphs, there is no general criterion for permanence. Theorem~\ref{thm:main} suggests that requiring the Hofbauer condition (Equation~\ref{eq:hofbauer:condition}) for each connected component of the invasion graph should be sufficient (see, e.g., the conjecture in \citet{spaak-etal-2023-assembly}). However, this has yet to be proved. Even if such a result holds, a substantial gap will remain. Specifically, there is a gap between the sufficient conditions for permanence in Theorem~\ref{thm:main} and the sufficient condition in Proposition~\ref{prop:extinction} for the existence of extinction-bound trajectories. This gap arises due to the potential formation of complex networks of heteroclinic orbits (trajectories connecting different invariant sets) in ecological models~\citep{hofbauer-1994-heteroclinic,schreiber-rittenhouse-2004-assembly,hofbauer-schreiber-2004-persist,hofbauer-2007-persist,spaak-etal-2023-assembly}.

Beyond these finite-dimensional models, there is a need to develop a more comprehensive approach based on invasion growth rates for infinite-dimensional impulsive systems. These models are widely used in mathematical ecology and epidemiology~\citep{luo-wang-2021-impulsive,fazly-etal-2017-impulsive,yang-zhong-2014-impulsive,lewis-li-2012-spreading}. Part of the challenge is the development of the appropriate spectral theory for non-autonomous operators. However, I speculate that the use of auxiliary variables may provide a useful way to circumvent some of these challenges.

\subsection{Biological Insights and Opportunities}

Applying the invasion growth rate criterion to three ecological models provides mathematically rigorous extensions of previously observed phenomena, offers several new biological insights, and highlights opportunities for future research.

The analysis of \citet{yi-dean-2013-bounded}'s model of microbial serial transfer experiments demonstrates how state-dependent timing of kicks (dilution to new containers) can act as a coexistence mechanism. In this model, competing species grow exponentially between dilution events in environments that alternately favor one species or the other. As observed by \citet{yi-dean-2013-bounded}, if dilutions occur at regularly spaced time intervals, competitive exclusion results: the species with the highest average per-capita growth rate excludes the other species {as $r_i(\mu)= \frac{\tau}{2}\sum_y \rho_i(y)+\ln \delta$}.

In contrast, when the timing of the dilution depends on the critical total population size ($x_1+x_2$), coexistence becomes possible through a storage effect mechanism~\citep[page E55]{li-chesson-2016-plankton}. This occurs because faster-growing species in any given environment reach higher densities more quickly, triggering earlier dilution events. These earlier dilutions effectively reduce the competitive advantage of the temporarily favored species, allowing both species to persist. This represents a classic fluctuation-dependent coexistence mechanism~\citep{chesson-1994-competition}.

The analysis presented here demonstrates that these conclusions hold even when the environment exhibits greater complexity in its fluctuations, e.g., with more than two environmental states or with quasiperiodic environments. More importantly, my analysis suggests that there are only two possible ecological outcomes: coexistence or exclusion. In particular, bistability, as observed in classical Lotka-Volterra competition models, is not possible, a conclusion not observed by \citet{yi-dean-2013-bounded}. However, verifying that exclusion occurs for all positive initial conditions remains an open mathematical challenge.

The analysis of a two-patch version of \citet{pachepsky-etal-2008-reproduction}'s consumer-resource model reveals a counterintuitive effect of reproductive delays on consumer persistence. \citet{pachepsky-etal-2008-reproduction} likely didn't comment on this effect due to their non-dimensionalization of time in the model. To understand this phenomenon, first consider the simpler case of a spatially homogeneous (or single patch) model. Here, the consumer invasion growth rate depends on two competing factors. First, the probability of survival to reproduction ($e^{-m\tau}$) decreases exponentially with longer delays ($\tau$). Second, reproductive output ($1+\theta\tau a \alpha/\beta$) increases only linearly with longer delays. This creates a fundamental trade-off: longer reproductive delays allow more time for offspring production, but the survival probability decreases more rapidly than offspring production increases. As a result of this trade-off, when the rate $\theta a \alpha/\beta$ of potential offspring production exceeds the per-capita mortality rate $m$, the per-capita growth rate is maximized at an intermediate reproductive delay of $\frac{1}{m}-\frac{1}{\theta a \alpha/\beta}$. {How spatial heterogeneity determines the} optimal reproductive delay is an interesting question that deserves further analysis.

Finally, the analysis of the empirically parameterized Lotka-Volterra model reveals a counterintuitive relationship between facilitation and disturbance effects on community composition. When competitors facilitate each other's growth (directly or indirectly), disturbance targeting one species can cascade through these positive interactions, potentially causing extinction of species {least} affected {directly} by the disturbance. This presents a paradox: the very facilitative interactions that might promote diversity under stable conditions can amplify vulnerability when disturbances occur. Given increasing empirical evidence that facilitation is more common in plant communities than previously thought~\citep{burns-strauss-2011-ecological,siefert-etal-2018-competition,buche-etal-2025-interactions}, using the methods presented here to explore the interaction between facilitation and disturbance in community composition provides exciting opportunities for future research. 

On the more mathematical side, our analysis of the Lotka-Volterra systems showed that the invasion graph for a flow-kick system is always a subgraph of the invasion graph for an appropriately defined autonomous Lotka-Volterra model. Understanding whether there are cases where it is a proper subgraph remains an interesting yet challenging mathematical problem.

\subsection{Concluding remarks}

In conclusion, this work advances both the mathematical theory of permanence and our ecological understanding of how continuous and discrete processes interact to influence species persistence. The invasion growth rate framework developed here provides a powerful tool for analyzing coexistence and extinction that complements existing approaches in theoretical ecology. Although important mathematical challenges remain—-particularly for systems with cyclic invasion graphs—the applications presented demonstrate the framework's ability to reveal counterintuitive ecological dynamics. These range from how state-dependent disturbances can promote coexistence in microbial communities to how facilitation networks can amplify disturbance effects in plant communities.

Future work that incorporates stochasticity, develops computational methods to calculate invasion graphs in a mathematically rigorous manner, and expands applications to new ecological contexts will further enhance our ability to predict and manage the dynamics of natural populations experiencing both continuous growth and discrete disturbances.

\section*{Acknowledgments} This work was partially funded by National Science Foundation Grant DEB-2243076.  The author thanks Jakob Karre-Rasmussen for constructive feedback on earlier versions of the manuscript and the participants of the joint Hastings-Schreiber Lab Meetings in Spring 2024 for inspiring work on flow-kick dynamics.

\bibliography{IG}

\begin{thebibliography}{70}
\providecommand{\natexlab}[1]{#1}
\providecommand{\url}[1]{\texttt{#1}}
\expandafter\ifx\csname urlstyle\endcsname\relax
  \providecommand{\doi}[1]{doi: #1}\else
  \providecommand{\doi}{doi: \begingroup \urlstyle{rm}\Url}\fi

\bibitem[Agur et~al.(1993)Agur, Cojocaru, Mazor, Anderson, and
  Danon]{agur-etal-1993-vaccination}
Z.~Agur, L.~Cojocaru, G.~Mazor, R.~M. Anderson, and Y.~L. Danon.
\newblock Pulse mass measles vaccination across age cohorts.
\newblock \emph{Proceedings of the National Academy of Sciences}, 90\penalty0
  (24):\penalty0 11698--11702, 1993.

\bibitem[Altenberg et~al.(2017)Altenberg, Liberman, and
  Feldman]{altenberg-etal-2017-reduction}
L.~Altenberg, U.~Liberman, and M.~W. Feldman.
\newblock Unified reduction principle for the evolution of mutation, migration,
  and recombination.
\newblock \emph{Proceedings of the National Academy of Sciences}, 114\penalty0
  (12):\penalty0 E2392--E2400, 2017.

\bibitem[Ashwin et~al.(1994)Ashwin, Buescu, and
  Stewart]{ashwin-etal-1994-attractor}
P.~Ashwin, J.~Buescu, and I.~N. Stewart.
\newblock From attractor to chaotic saddle: a tale of transverse stability.
\newblock \emph{Nonlinearity}, 9:\penalty0 703--737, 1994.

\bibitem[Bai(2015)]{bai-2015-seirs}
Z.~Bai.
\newblock Threshold dynamics of a time-delayed seirs model with pulse
  vaccination.
\newblock \emph{Mathematical biosciences}, 269:\penalty0 178--185, 2015.

\bibitem[Barab{\'a}s et~al.(2018)Barab{\'a}s, D'Andrea, and
  Stump]{barabas-etal-2018-chesson}
G.~Barab{\'a}s, R.~D'Andrea, and S.~M. Stump.
\newblock Chesson's coexistence theory.
\newblock \emph{Ecological Monographs}, 88:\penalty0 277--303, 2018.

\bibitem[Box(1979)]{box-1979-robustness}
G.~E.~P. Box.
\newblock Robustness in the strategy of scientific model building.
\newblock In R.~L. Launer and G.~N. Wilkinson, editors, \emph{Robustness in
  Statistics}, pages 201--236. Academic Press, 1979.

\bibitem[Brouwer(1911)]{brouwer-1911-uber}
L.~E.~J. Brouwer.
\newblock \"{U}ber {A}bbildungen von {M}annigfaltigkeiten.
\newblock \emph{Mathematische Annalen}, 71:\penalty0 97--115, 1911.
\newblock \doi{10.1007/BF01456931}.

\bibitem[Buche et~al.(2025)Buche, Shoemaker, Hallett, Bartomeus, Vesk,
  Weiss-Lehman, Mayfield, and Godoy]{buche-etal-2025-interactions}
L.~Buche, L.~G. Shoemaker, L.~M. Hallett, I.~Bartomeus, P.~Vesk,
  C.~Weiss-Lehman, M.~Mayfield, and O.~Godoy.
\newblock A continuum from positive to negative interactions drives plant
  species' performance in a diverse community.
\newblock \emph{Ecology Letters}, 28\penalty0 (1):\penalty0 e70059, 2025.

\bibitem[Burns and Strauss(2011)]{burns-strauss-2011-ecological}
J.~H. Burns and S.~Y. Strauss.
\newblock More closely related species are more ecologically similar in an
  experimental test.
\newblock \emph{Proceedings of the National Academy of Sciences}, 108\penalty0
  (13):\penalty0 5302--5307, 2011.

\bibitem[Butler et~al.(1986)Butler, Freedman, and
  Waltman]{butler-etal-1986-persistence}
G.~J. Butler, H.~I. Freedman, and P.~Waltman.
\newblock Uniformly persistent systems.
\newblock \emph{Proceedings of the American Mathematical Society}, 96:\penalty0
  425--430, 1986.

\bibitem[Chesson(1994)]{chesson-1994-competition}
P.~Chesson.
\newblock Multispecies competition in variable environments.
\newblock \emph{Theoretical population biology}, 45\penalty0 (3):\penalty0
  227--276, 1994.

\bibitem[Chesson(1982)]{chesson-1982-environment}
P.~L. Chesson.
\newblock The stabilizing effect of a random environment.
\newblock \emph{Journal of Mathematical Biology}, 15:\penalty0 1--36, 1982.

\bibitem[Clark et~al.(2024)Clark, Shoemaker, Arnoldi, Barab{\'a}s, Germain,
  Godoy, Hallett, Karako{\c{c}}, Saavedra, and
  Schreiber]{clark-etal-2024-coexistence}
A.~T. Clark, L.~G. Shoemaker, J.-F. Arnoldi, G.~Barab{\'a}s, R.~Germain,
  O.~Godoy, L.~Hallett, C.~Karako{\c{c}}, S.~Saavedra, and S.~Schreiber.
\newblock A practical guide to quantifying ecological coexistence.
\newblock \emph{EcoEvoRxv}, 2024.

\bibitem[Conley(1978)]{conley-1978-morse}
C.~Conley.
\newblock Isolated {I}nvariant {S}ets and {M}orse {I}ndex.
\newblock \emph{American Mathematical Society, CBMS}, 38, 1978.

\bibitem[d'Onofrio(2002)]{donofrio-2002-seir}
A.~d'Onofrio.
\newblock Stability properties of pulse vaccination strategy in seir epidemic
  model.
\newblock \emph{Mathematical biosciences}, 179\penalty0 (1):\penalty0 57--72,
  2002.

\bibitem[Ellner et~al.(2016)Ellner, Snyder, and
  Adler]{ellner-etal-2016-storage}
S.~P. Ellner, R.~E. Snyder, and P.~B. Adler.
\newblock How to quantify the temporal storage effect using simulations instead
  of math.
\newblock \emph{Ecology Letters}, 19:\penalty0 1333--1342, 2016.

\bibitem[Ellner et~al.(2019)Ellner, Snyder, Adler, and
  Hooker]{ellner-etal-2019-coexistence}
S.~P. Ellner, R.~E. Snyder, P.~B. Adler, and G.~Hooker.
\newblock An expanded modern coexistence theory for empirical applications.
\newblock \emph{Ecology letters}, 22\penalty0 (1):\penalty0 3--18, 2019.

\bibitem[Fazly et~al.(2017)Fazly, Lewis, and Wang]{fazly-etal-2017-impulsive}
M.~Fazly, M.~Lewis, and H.~Wang.
\newblock On impulsive reaction-diffusion models in higher dimensions.
\newblock \emph{SIAM Journal on Applied Mathematics}, 77\penalty0 (1):\penalty0
  224--246, 2017.

\bibitem[Feldman and Liberman(1986)]{feldman-liberman-1986-modifiers}
M.~W. Feldman and U.~Liberman.
\newblock An evolutionary reduction principle for genetic modifiers.
\newblock \emph{Proceedings of the National Academy of Sciences}, 83\penalty0
  (13):\penalty0 4824--4827, 1986.

\bibitem[Garay(1989)]{garay-1989-persistence}
B.~M. Garay.
\newblock Uniform persistence and chain recurrence.
\newblock \emph{Journal of mathematical analysis and applications},
  139\penalty0 (2):\penalty0 372--381, 1989.

\bibitem[Garay and Hofbauer(2003)]{garay-hofbauer-2003-permanence}
B.~M. Garay and J.~Hofbauer.
\newblock Robust permanence for ecological differential equations, minimax, and
  discretizations.
\newblock \emph{SIAM Journal on Mathematical Analysis}, 34\penalty0
  (5):\penalty0 1007--1039, 2003.

\bibitem[Geijzendorffer et~al.(2011)Geijzendorffer, van~der Werf, Bianchi, and
  Schulte]{geijzendorffer-etal-2011-sustained}
I.~Geijzendorffer, W.~van~der Werf, F.~Bianchi, and R.~Schulte.
\newblock Sustained dynamic transience in a lotka--volterra competition model
  system for grassland species.
\newblock \emph{Ecological Modelling}, 222\penalty0 (17):\penalty0 2817--2824,
  2011.

\bibitem[Geng and Lutscher(2021)]{geng-lutscher-2021-seasonal}
Y.~Geng and F.~Lutscher.
\newblock Competitive coexistence of seasonal breeders.
\newblock \emph{Journal of Mathematical Biology}, 83:\penalty0 1--35, 2021.

\bibitem[Geng et~al.(2021)Geng, Wang, and Lutscher]{geng-etal-2021-coexistence}
Y.~Geng, X.~Wang, and F.~Lutscher.
\newblock Coexistence of competing consumers on a single resource in a hybrid
  model.
\newblock \emph{Discrete \& Continuous Dynamical Systems-Series B}, 26\penalty0
  (1), 2021.

\bibitem[Good et~al.(2017)Good, McDonald, Barrick, Lenski, and
  Desai]{good-etal-2017-molecular}
B.~H. Good, M.~J. McDonald, J.~E. Barrick, R.~E. Lenski, and M.~M. Desai.
\newblock The dynamics of molecular evolution over 60,000 generations.
\newblock \emph{Nature}, 551\penalty0 (7678):\penalty0 45--50, 2017.

\bibitem[Gyllenberg et~al.(1997)Gyllenberg, Hanski, and
  Lindstr{\"o}m]{gyllenberg-etal-1997-continuous}
M.~Gyllenberg, I.~Hanski, and T.~Lindstr{\"o}m.
\newblock Continuous versus discrete single species population models with
  adjustable reproductive strategies.
\newblock \emph{Bulletin of Mathematical Biology}, 59\penalty0 (4):\penalty0
  679--705, 1997.

\bibitem[Hofbauer(1981)]{hofbauer-1981-hypercycles}
J.~Hofbauer.
\newblock A general cooperation theorem for hypercycles.
\newblock \emph{Monatshefte f{\"u}r Mathematik}, 91:\penalty0 233--240, 1981.

\bibitem[Hofbauer(1994)]{hofbauer-1994-heteroclinic}
J.~Hofbauer.
\newblock Heteroclinic cycles in ecological differential equations.
\newblock \emph{Equadiff 8}, pages 105--116, 1994.

\bibitem[Hofbauer(2007)]{hofbauer-2007-persist}
J.~Hofbauer.
\newblock To persist or not to persist—differential equations in ecology.
\newblock \emph{Trends in dynamical systems}, pages 19--32, 2007.

\bibitem[Hofbauer and Schreiber(2004)]{hofbauer-schreiber-2004-persist}
J.~Hofbauer and S.~J. Schreiber.
\newblock To persist or not to persist?
\newblock \emph{Nonlinearity}, 17\penalty0 (4):\penalty0 1393, 2004.

\bibitem[Hofbauer and Schreiber(2010)]{hofbauer-schreiber-2010-robust}
J.~Hofbauer and S.~J. Schreiber.
\newblock Robust permanence for interacting structured populations.
\newblock \emph{Journal of Differential Equations}, 248:\penalty0 1955--1971,
  2010.

\bibitem[Hofbauer and Schreiber(2022)]{hofbauer-schreiber-2022-permanence}
J.~Hofbauer and S.~J. Schreiber.
\newblock Permanence via invasion graphs: {I}ncorporating community assembly
  into modern coexistence theory.
\newblock \emph{Journal of Mathematical Biology}, 85\penalty0 (5):\penalty0 54,
  2022.

\bibitem[Hofbauer and So(1989)]{hofbauer-so-1989-persistence}
J.~Hofbauer and J.~W.~H. So.
\newblock Uniform persistence and repellors for maps.
\newblock \emph{Proceedings of the American Mathematical Society},
  107:\penalty0 1137--1142, 1989.

\bibitem[Horn and Johnson(1994)]{horn-johnson-1994-matrix}
R.~A. Horn and C.~R. Johnson.
\newblock \emph{Topics in matrix analysis}.
\newblock Cambridge university press, 1994.

\bibitem[Huston(1979)]{huston-1979-diversity}
M.~Huston.
\newblock A general hypothesis of species diversity.
\newblock \emph{The American Naturalist}, 113\penalty0 (1):\penalty0 81--101,
  1979.
\newblock \doi{10.1086/283367}.

\bibitem[Hutson and Schmitt(1992)]{hutson-schmitt-1992-permanence}
V.~Hutson and K.~Schmitt.
\newblock Permanence and the dynamics of biological systems.
\newblock \emph{Mathematical Biosciences}, 111:\penalty0 1--71, 1992.

\bibitem[Jansen and Sigmund(1998)]{jansen-sigmund-1998-permanence}
V.~A.~A. Jansen and K.~Sigmund.
\newblock Shaken not stirred: On permanence in ecological communities.
\newblock \emph{Theoretical Population Biology}, 54:\penalty0 195--201, 1998.

\bibitem[Katok et~al.(1995)Katok, Katok, and
  Hasselblatt]{katok-etal-1995-dynamical-systems}
A.~Katok, A.~Katok, and B.~Hasselblatt.
\newblock \emph{Introduction to the modern theory of dynamical systems}.
\newblock Number~54. Cambridge University Press, 1995.

\bibitem[Kirkland et~al.(2006)Kirkland, Li, and
  Schreiber]{kirkland-etal-2006-dispersal}
S.~Kirkland, C.-K. Li, and S.~J. Schreiber.
\newblock On the evolution of dispersal in patchy landscapes.
\newblock \emph{SIAM Journal on Applied Mathematics}, 66\penalty0 (4):\penalty0
  1366--1382, 2006.

\bibitem[Lakshmikantham et~al.(1989)Lakshmikantham, Bainov, and
  Simeonov]{lakshmikantham-etal-1989-impulsive}
V.~Lakshmikantham, D.~Bainov, and P.~Simeonov.
\newblock \emph{Theory of Impulsive Differential Equations}.
\newblock World Scientific, Singapore, 1989.

\bibitem[Letten and Ludington(2023)]{letten-ludington-2023-pulsed}
A.~D. Letten and W.~B. Ludington.
\newblock Pulsed, continuous or somewhere in between? resource dynamics matter
  in the optimisation of microbial communities.
\newblock \emph{The ISME Journal}, 17\penalty0 (4):\penalty0 641--644, 2023.

\bibitem[Lewis and Li(2012)]{lewis-li-2012-spreading}
M.~A. Lewis and B.~Li.
\newblock Spreading speed, traveling waves, and minimal domain size in
  impulsive reaction--diffusion models.
\newblock \emph{Bulletin of mathematical biology}, 74:\penalty0 2383--2402,
  2012.

\bibitem[Li and Chesson(2016)]{li-chesson-2016-plankton}
L.~Li and P.~Chesson.
\newblock The effects of dynamical rates on species coexistence in a variable
  environment: the paradox of the plankton revisited.
\newblock \emph{The American Naturalist}, 188\penalty0 (2):\penalty0 E46--E58,
  2016.

\bibitem[Luo and Wang(2021)]{luo-wang-2021-impulsive}
D.~Luo and Q.~Wang.
\newblock Dynamic analysis on an almost periodic predator-prey system with
  impulsive effects and time delays.
\newblock \emph{Discrete \& Continuous Dynamical Systems-Series B}, 26\penalty0
  (6), 2021.

\bibitem[MacArthur and Levins(1967)]{macarthur-levins-1967-limiting}
R.~MacArthur and R.~Levins.
\newblock The limiting similarity, convergence, and divergence of coexisting
  species.
\newblock \emph{The American Naturalist}, 101:\penalty0 377--385, 1967.

\bibitem[Ma{\~{n}\'e}(1983)]{mane-1983-ergodic}
R.~Ma{\~{n}\'e}.
\newblock \emph{Ergodic Theory and Differentiable Dynamics}.
\newblock Springer-Verlag, New York, 1983.

\bibitem[McGehee and Armstrong(1977)]{mcgehee-armstrong-1977-exclusion}
R.~McGehee and R.~A. Armstrong.
\newblock Some mathematical problems concerning the ecological principle of
  competitive exclusion.
\newblock \emph{J. Differential Equations}, 23:\penalty0 30--52, 1977.

\bibitem[Meng and Chen(2008)]{meng-chen-2008-sir-pulse}
X.~Meng and L.~Chen.
\newblock The dynamics of a new sir epidemic model concerning pulse vaccination
  strategy.
\newblock \emph{Applied Mathematics and Computation}, 197\penalty0
  (2):\penalty0 582--597, 2008.

\bibitem[Meyer et~al.(2018)Meyer, Hoyer-Leitzel, Iams, Klasky, Lee, Ligtenberg,
  Bussmann, and Zeeman]{meyer-etal-2018-resilience}
K.~Meyer, A.~Hoyer-Leitzel, S.~Iams, I.~Klasky, V.~Lee, S.~Ligtenberg,
  E.~Bussmann, and M.~L. Zeeman.
\newblock Quantifying resilience to recurrent ecosystem disturbances using
  flow--kick dynamics.
\newblock \emph{Nature Sustainability}, 1\penalty0 (11):\penalty0 671--678,
  2018.
\newblock \doi{10.1038/s41893-018-0168-z}.

\bibitem[Pachepsky et~al.(2008)Pachepsky, Nisbet, and
  Murdoch]{pachepsky-etal-2008-reproduction}
E.~Pachepsky, R.~M. Nisbet, and W.~W. Murdoch.
\newblock Between discrete and continuous: consumer--resource dynamics with
  synchronized reproduction.
\newblock \emph{Ecology}, 89\penalty0 (1):\penalty0 280--288, 2008.

\bibitem[Patel and Schreiber(2018)]{patel-schreiber-2018-permanence}
S.~Patel and S.~J. Schreiber.
\newblock Robust permanence for ecological equations with internal and external
  feedbacks.
\newblock \emph{Journal of Mathematical Biology}, 77:\penalty0 79--105, 2018.

\bibitem[Perko(2001)]{perko2001differential}
L.~Perko.
\newblock \emph{Differential Equations and Dynamical Systems}.
\newblock Texts in Applied Mathematics. Springer, New York, NY, 3 edition,
  2001.
\newblock ISBN 978-0-387-95116-4.
\newblock \doi{10.1007/978-1-4613-0003-8}.

\bibitem[Rand et~al.(1994)Rand, Wilson, and McGlade]{rand-etal-1994-evolution}
D.~A. Rand, H.~Wilson, and J.~M. McGlade.
\newblock Dynamics and evolution: evolutionarily stable attractors, invasion
  exponents and phenotype dynamics.
\newblock \emph{Philosophical Transactions of the Royal Society of London.
  Series B: Biological Sciences}, 343\penalty0 (1305):\penalty0 261--283, 1994.

\bibitem[Roth et~al.(2017)Roth, Salceanu, and
  Schreiber]{roth-etal-2017-permanence}
G.~Roth, P.~L. Salceanu, and S.~J. Schreiber.
\newblock Robust permanence for ecological maps.
\newblock \emph{SIAM Journal of Mathematical Analysis}, 49:\penalty0
  3527--3549, 2017.

\bibitem[Royden and Fitzpatrick(2010)]{royden-fitzpatrick-2010-real}
H.~L. Royden and P.~M. Fitzpatrick.
\newblock \emph{Real Analysis}.
\newblock Pearson, 4 edition, 2010.

\bibitem[Ruelle(1979)]{ruelle-1979-exponents}
D.~Ruelle.
\newblock Analycity properties of the characteristic exponents of random matrix
  products.
\newblock \emph{Advances in mathematics}, 32\penalty0 (1):\penalty0 68--80,
  1979.

\bibitem[Schreiber(2000)]{schreiber-2000-permanence}
S.~J. Schreiber.
\newblock Criteria for cr robust permanence.
\newblock \emph{Journal of Differential Equations}, 162\penalty0 (2):\penalty0
  400--426, 2000.

\bibitem[Schreiber(2006)]{schreiber-2006-persistence}
S.~J. Schreiber.
\newblock Persistence despite perturbations for interacting populations.
\newblock \emph{Journal of Theoretical Biology}, 242:\penalty0 844--852, 2006.

\bibitem[Schreiber and Rittenhouse(2004)]{schreiber-rittenhouse-2004-assembly}
S.~J. Schreiber and S.~Rittenhouse.
\newblock From simple rules to cycling in community assembly.
\newblock \emph{Oikos}, 105\penalty0 (2):\penalty0 349--358, 2004.

\bibitem[Schreiber et~al.(2011)Schreiber, Bena\"{i}m, and
  Atchad\'{e}]{schreiber-etal-2011-persistence}
S.~J. Schreiber, M.~Bena\"{i}m, and K.~A.~S. Atchad\'{e}.
\newblock Persistence in fluctuating environments.
\newblock \emph{Journal of Mathematical Biology}, 62:\penalty0 655--683, 2011.

\bibitem[Schuster et~al.(1979)Schuster, Sigmund, and
  Wolff]{schuster-etal-1979-hypercycles}
P.~Schuster, K.~Sigmund, and R.~Wolff.
\newblock Dynamical systems under constant organization 3: Cooperative and
  competitive behavior of hypercycles.
\newblock \emph{Journal of Differential Equations}, 32:\penalty0 357--368,
  1979.

\bibitem[Siefert et~al.(2018)Siefert, Zillig, Friesen, and
  Strauss]{siefert-etal-2018-competition}
A.~Siefert, K.~W. Zillig, M.~L. Friesen, and S.~Y. Strauss.
\newblock Soil microbial communities alter conspecific and congeneric
  competition consistent with patterns of field coexistence in three trifolium
  congeners.
\newblock \emph{Journal of Ecology}, 106\penalty0 (5):\penalty0 1876--1891,
  2018.

\bibitem[Sigmund and Schuster(1984)]{sigmund-schuster-1984-permanence}
K.~Sigmund and P.~Schuster.
\newblock Permanence and uninvadability for deterministic population models.
\newblock In \emph{Stochastic phenomena and chaotic behaviour in complex
  systems}, pages 173--184. Springer, 1984.

\bibitem[Spaak and Schreiber(2023)]{spaak-schreiber-2023-assembly}
J.~W. Spaak and S.~J. Schreiber.
\newblock Building modern coexistence theory from the ground up: the role of
  community assembly.
\newblock \emph{Ecology Letters}, 26\penalty0 (11):\penalty0 1840--1861, 2023.

\bibitem[Spaak et~al.(2023)Spaak, Adler, and Ellner]{spaak-etal-2023-assembly}
J.~W. Spaak, P.~B. Adler, and S.~P. Ellner.
\newblock Continuous assembly required: Perpetual species turnover in
  two-trophic-level ecosystems.
\newblock \emph{Ecosphere}, 14\penalty0 (7):\penalty0 e4614, 2023.

\bibitem[Venkataram et~al.(2016)Venkataram, Dunn, Li, Agarwala, Chang, Ebel,
  Geiler-Samerotte, Herissant, Blundell, Levy,
  et~al.]{venkataram-etal-2016-fitness}
S.~Venkataram, B.~Dunn, Y.~Li, A.~Agarwala, J.~Chang, E.~R. Ebel,
  K.~Geiler-Samerotte, L.~Herissant, J.~R. Blundell, S.~F. Levy, et~al.
\newblock Development of a comprehensive genotype-to-fitness map of
  adaptation-driving mutations in yeast.
\newblock \emph{Cell}, 166\penalty0 (6):\penalty0 1585--1596, 2016.

\bibitem[Wolfe and Dutton(2015)]{wolfe-dutton-2015-fermented}
B.~E. Wolfe and R.~J. Dutton.
\newblock Fermented foods as experimentally tractable microbial ecosystems.
\newblock \emph{Cell}, 161\penalty0 (1):\penalty0 49--55, 2015.

\bibitem[Yang and Zhong(2014)]{yang-zhong-2014-impulsive}
L.~Yang and S.~Zhong.
\newblock Dynamics of a delayed stage-structured model with impulsive
  harvesting and diffusion.
\newblock \emph{Ecological Complexity}, 19:\penalty0 111--123, 2014.

\bibitem[Yi and Dean(2013)]{yi-dean-2013-bounded}
X.~Yi and A.~M. Dean.
\newblock Bounded population sizes, fluctuating selection and the tempo and
  mode of coexistence.
\newblock \emph{Proceedings of the National Academy of Sciences}, 110\penalty0
  (42):\penalty0 16945--16950, 2013.

\bibitem[Yurtsev et~al.(2016)Yurtsev, Conwill, and
  Gore]{yurtsev-etal-2016-oscillatory}
E.~A. Yurtsev, A.~Conwill, and J.~Gore.
\newblock Oscillatory dynamics in a bacterial cross-protection mutualism.
\newblock \emph{Proceedings of the National Academy of Sciences}, 113\penalty0
  (22):\penalty0 6236--6241, 2016.

\end{thebibliography}

\section{Appendix A: Proofs of Theorem~\ref{thm:main} and Lemma~\ref{lem:kick-to-kick-permanence-implies-flow-kick-permanence}}\label{proof:main_theorem}
To prove Theorem~\ref{thm:main}, I use a characterization of permanence for maps due to \citet{hofbauer-so-1989-persistence}. To state this characterization, recall that the stable set $W^s(A)$ for an invariant set $A\subset \S$ for $\kappa$ is 
\[
W^s(A)=\{z\in \S: \omega(z)\subset A\}.
\]
Let $\{M_1,\dots,M_\ell\}$ be a Morse decomposition of $\Gamma_0$. 

\begin{theorem}[\citet{hofbauer-so-1989-persistence}]\label{thm:hofbauer-so} Assume each Morse set $M_i$ is isolated with respect to $\S$ i.e. $M_i$ is the largest invariant set in some neighborhood $U_i\subset \S$ of $M_i$. Then \eqref{eq:flow-kick} is permanent if and only if $W^s(M_i)\subset \S_0$ for all $i$
\end{theorem}

To use Theorem~\ref{thm:hofbauer-so}, one needs to verify two properties for each Morse set $M_i$: (i) $W^s(M_i)\subset \S_0$ and (ii) $M_i$ is isolated in $\S$. Let $M_i$ be a Morse set.   Lemma~\ref{lem:key:one} holds the key to verifying both properties.   {The proof of this lemma relies on weak* convergence of Borel probability measures on a compact set $\Gamma$~\citep{royden-fitzpatrick-2010-real} . Specifically, a sequence of Borel probability  measures $\mu_n$ with $\mu_n(\Gamma)=1$  converges in the weak* topology to the Borel probability measure $\mu$  with $\mu(\Gamma)=1$ (i.e. $\lim_{n\to\infty}\mu_n =\mu$) if $\lim_{n\to\infty}\int h(z)\mu_n(dz)=\int h(z)\mu(dz)$ for all continuous functions functions $h:\Gamma\to\R.$}

\begin{proof}[Proof of Lemma~\ref{lem:key:one}.]
Let $\tilde z\in \S_+$. Define a sequence of probability measures $\mu_n$ by 
\[
\int h(z)\mu_{n}(dz)=\frac{1}{n}\sum_{m=1}^{n} h(\kappa^m(\tilde z)) \mbox{ for all continuous functions }h:\S\to\R 
\]
i.e. the time average of the forward orbit of $\tilde z$ until time $n$. By assumption \textbf{A5}
and weak* compactness, there is a subsequence $n_\ell \to \infty$ such that $\mu_{n_\ell}$ converges in the weak* topology to  a Borel probability measure $\mu$ as $\ell\to\infty$ {i.e. \begin{equation}\label{eq:weak*}\lim_{\ell\to \infty} \int h(z)\mu_{n_\ell}(dz)=\lim_{\ell\to \infty} \frac{1}{n_\ell}\sum_{m=1}^{n_\ell} h(\kappa^m(\tilde z)) =\int h(z)\mu(dz) \mbox{ for all continuous functions }h:\S\to\R.\end{equation}} By the definition of $\mu_n$ and weak* convergence, $\mu(\omega(\tilde z))=1$. By assumption \textbf{A5}, $\omega(\tilde z)\subset \Gamma$. I claim that $\mu$ is an invariant measure for the kick-to-kick map $\kappa$. To see why, it suffices to verify that $\int h(\kappa(z))\mu(dz) = \int h(z)\mu(dz)$ for any continuous function $h: \Gamma \to \R$. Let $h: \Gamma \to \R$ be a continuous function. Then
\[
\begin{aligned}
\left| \int (h(\kappa(z)) - h(z)) \mu(dz) \right| &=& \lim_{\ell \to \infty} \frac{1}{n_\ell} \left| \sum_{m=1}^{n_\ell} h(\kappa^{m+1}(\tilde z))-h(\kappa^m(\tilde z)) \right|\\
&=& \lim_{\ell \to \infty} \frac{1}{n_\ell} \left| h(\kappa^{n_\ell+1}(\tilde z))-h(\kappa(\tilde z)) \right|\\
& \le & \limsup_{\ell \to \infty} \frac{1}{n_\ell} 2\max_{z\in \Gamma}|h(z)|=0
\end{aligned}
\]
As this holds for any continuous $h:\Gamma\to \R$, $\mu$ is an invariant measure. 

For each $i\in \{1,2,\dots,k\}$ and $z=(x,y)$, let $\pi_i(z)=x_i$ be the projection onto the density of species $i$. {As $\int_0^t \frac{c'(t)}{c(t)}dt=\ln \frac{c(t)}{c(0)}$ for any continuously, differentiable, positive function $c:\R\to\R$, it follows that f}or any initial condition $z$ with $\pi_i(z)>0$
\[
\ln \frac{\pi_i(z.t)}{\pi_i(z)}=\int_0^t f_i(z.s)ds
\]
and, after exponentiating,  
\begin{equation}\label{eq:percapita_representation}
\pi_i(z.t)=\pi_i(z)\exp\left(\int_0^t f_i(z.s)ds\right).
\end{equation}
Equation~\eqref{eq:percapita_representation} holds trivially when $\pi_i(z)=0$.
Thus, the kick-to-kick map satisfies $\pi_i(\kappa (z))=\pi_i(z) \Psi_i(z)$ where 
\begin{equation}\label{eq:fitness}
\Psi_i(z)=\exp\left(\int_0^{\tau(z)} f_i(z.s) \,ds\right) F_i(z.\tau(z)).
\end{equation}
Using this representation recursively, one gets
\[
\pi_i\left(\kappa^n(z)\right)=\pi_i(z) \prod_{m =0}^{n-1} \Psi_i\left(\kappa^m(z)\right).
\]
and when $\pi_i(z)>0$,
\[
\ln \frac{\pi_i\left(\kappa^n(z)\right)}{\pi_i (z)}=\sum_{m =0}^{n-1} \ln \Psi_i\left(\kappa^m(z)\right).
\]
{Equation~\ref{eq:weak*}, the definition of $\Psi_i$, and the definition of $r_i(\mu)$ in equations~\eqref{eq:per-capita_growth_rate}-\eqref{eq:per-capita_growth_rate_at_mu} imply}
\[
r_i(\mu)=\lim_{\ell\to\infty}\frac{1}{n_\ell}\sum_{m =0}^{n_\ell-1} \ln \Psi_i\left(\kappa^m(\tilde z)\right).
\]
Assumption \textbf{A3} implies 
\begin{equation}\label{eq:ri_upperbound}
\lim_{\ell\to\infty}\frac{1}{n_\ell}\sum_{m =0}^{n_\ell-1} \ln \Psi_i\left(\kappa^m(\tilde z)\right)=\lim_{\ell\to\infty}\ln \frac{\pi_i\left(\kappa^{n_\ell}(\tilde z)\right)}{\pi_i (\tilde z)}\le \limsup_{\ell\to\infty}\ln \frac{K}{\pi_i (\tilde z)} = 0
\end{equation}
where $K$ comes from assumption \textbf{A5}. Thus, $r_i(\mu)\le 0$ for all $i\in\{1,2,\dots,k\}$.

Now, let $\mathcal{O}^-=\{z(-n)\}_{n=0}^\infty\subset \S_+\cap \Gamma$ be a negative orbit for the kick-to-kick map, i.e., $z(-n+1)=\kappa(z(-n))$ for all $n\le 1$. For $n\ge 1$ define the probability measure $\mu_n$ by
\[
\int h(z)\mu_n(dz)=\frac{1}{n}\sum_{m=1-n}^{0} h(z(m)) \mbox{ for all continuous }h:\Gamma\to \R.
\]
Passing to a subsequence if necessary, let $\mu$ be a weak* limit point $\mu_n$ as $n\to \infty$. By construction, $\mu(\alpha(\mathcal{O}^-))=1$. I claim that $\mu$ is invariant for the kick-to-kick map. As in the case of the forward orbits, it suffices to show that 
\[
\int_\S h(z)\mu(dz)=\int_\S h(\kappa(z))\mu(dz) \mbox{ for all continuous }h:\Gamma\to \R.
\]
By weak* convergence
\[
\begin{aligned}
\left|\int_\S \left(h(z)- h(\kappa (z))\right)\mu(dz)\right|=&\lim_{n\to\infty}\frac{1}{n}\left|\sum_{m=1-n}^{0} h(z(m))- h(\kappa(z(m)))\right|\\
=&\lim_{n\to\infty}\frac{1}{n}\left|h(z(1-n))- h(\kappa(z(0)))\right|
= 0.
\end{aligned}
\]
As this holds for any continuous $h:\Gamma\to \R$,  $\mu$ is $\kappa$ invariant. By construction $\mu(\alpha(\mathcal{O}^-)=1$.

To prove $r_i(\mu)\le 0$ for all $i\in\{1,2,\dots,k\}$, one needs to slightly modify the argument used in the case of an $\omega$-limit set.  Let $\Psi_i(z)$ be defined by \eqref{eq:fitness}. As $\pi_i(\kappa(z))=\pi_i(z)\Psi_i(z)$, one has
\[
\pi_i(z(0))=\pi_i(z(-n))\prod_{m=-n}^{-1} \Psi_i(z(m))
\]
and if $\pi_i(z(-n))>0$,
\[
\ln\frac{ \pi_i(z(0))}{\pi_i(z(-n))}=\sum_{m=-n}^{-1} \ln \Psi_i(z(m)).
\]
By the definition of $\Psi_i$ and weak* convergence of the $\mu_n$, 
\[
r_i(\mu)=\lim_{n\to\infty}\frac{1}{n}\ln\frac{ \pi_i(z(0))}{\pi_i(z(-n))}
\ge \lim_{n\to\infty}\frac{1}{n}\ln\frac{\pi_i(z(0))}{K}=0
\]
for all $i\in \{1,2,\dots,k\}$ and where $K$ comes from assumption \textbf{A5}.  Thus, $r_i(\mu)\ge 0$ for all $i\in\{1,2,\dots,k\}$.
\end{proof}

To verify that the Morse sets $M_i$ are isolated in $\S$, I need the following lemma.

\begin{lemma}\label{lem:key:two} Let $M\subset \Gamma_0$ be a compact,  invariant set for $\kappa$ that is isolated with respect $\Gamma_0$. If $M$ is not isolated with respect to $\S$, then there exist invariant probability measures $\mu_+,\mu_-$ such that $r_i(\mu_+)\ge 0$ and $r_i(\mu_-)\le 0$ for all $i$  and $\mu_+(M)=\mu_-(M)=1$.
\end{lemma}

For a compact set $M$, define $\rm{dist}(z,M)=\max_{\tilde z\in M}\|z-\tilde z{\|}$.

\begin{proof}[Proof of Lemma~\ref{lem:key:two}]
    As $M$ is not isolated in $\S$, for every natural number $\ell$ there is a compact invariant set $K_\ell\subset \Gamma$ such that $K_\ell \setminus M\neq \emptyset$ and $K_\ell\subset\{z\in \Gamma: \rm{dist}(z,M)\le 1/\ell\}$. As $M$ is isolated with respect to $\Gamma_0$, $K_\ell\cap \Gamma_0=M$ which implies that $K_\ell \cap \Gamma_+\neq \emptyset$. 
    
    For each $\ell$, choose a point $z\in K_\ell \cap \Gamma_+$. As $K_\ell$ is invariant, $\omega(z)\subset K_\ell$.  Lemma~\ref{lem:key:one} implies there is an invariant measure $\mu_\ell$ such that $\mu_\ell(K_\ell)=1$ and $r_i(\mu_\ell)\le 0$ for all $i\in \{1,2,\dots,k\}$. Let $\mu$ be a weak* limit point of $\mu_\ell$ as $\ell\to\infty$. As the set of invariant probability measures supported by $\Gamma$ is compact with respect to the weak* topology, $\mu$ is an invariant measure. As $\mu(K_\ell)=1$ and $K_\ell\subset\{z\in \Gamma: \rm{dist}(z,M)\le 1/\ell\}$, $\mu(M)=1$. Weak* convergence implies that $r_i(\mu)\le 0$ for $i\in \{1,2,\dots,k\}$.

    Next, I prove the existence of an invariant probability measure $\mu$ such that $r_i(\mu)\ge 0$ for $i\in \{1,2,\dots,k\}$ and $\mu(M)=1$. For each $\ell$, choose $z\in K_\ell\cap \Gamma_+$ and a negative orbit $\mathcal{O}^-=\{z(n)\}_{n=-\infty}^0\subset \Gamma_+\cap K_\ell$ with $z(0)=z$. Lemma~\ref{lem:key:one} implies there is an invariant measure $\mu_\ell$ such that $\mu_\ell(K_\ell)=1$ and $r_i(\mu_\ell)\ge 0$ for all $i\in \{1,2,\dots,k\}$. Let $\mu$ be a weak* limit point of $\mu_\ell$ as $\ell\to\infty$. For the same reasoning as in the previous paragraph, $\mu$ is an invariant probability measure that satisfies $\mu(M)=1$ and $r_i(\mu)\ge 0$ for $i\in \{1,2,\dots,k\}$.
\end{proof}

Armed with these two lemmas,  the proof of Theorem~\ref{thm:main} is straightforward. 

\begin{proof}[Proof of Theorem~\ref{thm:main}.] Let $M_i$ be a Morse set. The first assertion of Lemma~\ref{lem:key:one} and the assumption that $\max_i r_i(\mu)>0$ for all invariant probability measures with $\mu(M_i)=1$ implies that $W^s(M_i)\subset \S_+$. Similarly, Lemma~\ref{lem:key:two} implies that $M_i$ is isolated with respect to $\S$. Applying Theorem~\ref{thm:hofbauer-so} implies that $\kappa$ is permanent. 
\end{proof}

To link the permanence of the kick-to-kick map to the permanence of the flow-kick dynamics, I conclude this section with a proof of lemma~\ref{lem:kick-to-kick-permanence-implies-flow-kick-permanence}. \vskip 0.1in

\begin{proof}[Proof of Lemma~\ref{lem:kick-to-kick-permanence-implies-flow-kick-permanence}.] Assumption \textbf{A5} implies that there exists $K>0$ is such that 
\[
\limsup_{n\to\infty}\|\kappa^n(z)\|\le K \mbox{ for all }z\in \S.
\]
Define 
\[
\widetilde K=\sup_{z\in \S, \|z\|\le 1.01 K} \{\|z.t\|:0\le t\le \tau(z)\} \times 1.01
\]
By compactness of $\{z\in \S: \|z\|\le 1.01 K\}$, the continuity of flow $(z,t)\to z.t$ and the continuity of $\tau$, one has $\widetilde K<\infty$. Given any $z\in \S$, choose $n^*$ (depending on $z$) such that 
\[
\|\kappa^n(z)\|\le 1.01 K \mbox{ for }n\ge n^*.
\]
By the definition of $\widetilde K$, 
\[
{\Phi}(t,z)\le \widetilde K \mbox{ for }t\ge \sum_{m=0}^{n^*}\tau(\kappa^m(z)). 
\]

Assume $M>0$ is such that 
\[
\liminf_{n\to\infty}\pi_i(\kappa^n(z))\ge M \mbox{ for all }z\in \S_+ \mbox{ and }i\in\{1,2,\dots,k\}.
\]
Define
\[
\rho=\inf \left\{ \int_0^t f_i(z.s)ds : z\in \S, \|z\|\le 1.01 K, i\in \{1,2,\dots,k\}, 0\le t\le \tau(z)\right\} -0.01
\]
By the compactness of $\{z\in \S:\|z\|\le 1.01K\}$, continuity of $\tau$ and the continuity of $f_i$, one has $\rho>-\infty$. Moreover, as $\int_0^0 f_i(z.s)ds=0$, one has $\rho\le -0.01$. Given any $z\in \S_+$, choose $n^*$ (depending on $z$) sufficiently large so that 
\[
\|\kappa^n(z)\|\le 1.01 K \mbox{ and }\min_i\pi_i(\kappa^n (z))\ge e^{-0.01} M \mbox{ for }n\ge n^*.
\]
Then, the definition of $\rho$ implies that 
\[
\min_i \pi_i({\Phi}(t,z)) \ge e^{\rho-0.01}  M \mbox{ for }t\ge \sum_{m=0}^{n^*}\tau(\kappa^m(z)).
\]
Choosing $\widetilde M =e^{\rho-0.01}  M$ completes the proof of the lemma. \end{proof}

\section{Appendix B: Proof of Theorem~\ref{thm:main2}}\label{proof:IG}

Throughout this proof,  assume that  \textbf{A1}--\textbf{A7} hold. For any $z=(x,y)$, define $\pi_i( z) =x_i$. I begin with a proof of Lemma~\ref{lemma:critical} that naturally follows from Lemma~\ref{lem:key:one}.

\begin{proof}[Proof of Lemma~\ref{lemma:critical}.] Assume $z\in \S_I$ for some $I\subset \{1,2,\dots,k\}$ and $\omega(z)\in \S_J$ for some $J\subset I$. Lemma~\ref{lem:key:one} implies there exists an invariant measure $\mu$ with compact support such that $\mu(\S_J)=1$ and $r_i(\mu)\le 0$ for all $i\in I$. Assumption \textbf{A6} implies that $r_i(\mu)<0$ for all $i\in I\setminus J$. 

Let $\mathcal{O}^-=\{z(n)\}_{n=-\infty}^0$ be a negative orbit for the kick-to-kick map i.e. $z(n+1)=\kappa(z(n))$ for all $n\le -1$.  Assume that $z(0)\in \S_I$ for some $I\subset\{1,2,\dots,k\}$ and $\alpha(\mathcal{O}^-)\subset \S_J$ for some $J\subset \{1,2,\dots,k\}$. Lemma~\ref{lem:key:one} implies there exists an invariant measure $\mu$ with compact support such that $\mu(\S_J)=1$ and $r_i(\mu)\ge 0$ for all $i\in I$. Assumption \textbf{A6} implies that $r_i(\mu)>0$ for all $i\in I\setminus J$. 
\end{proof}

The following lemma constructs a Morse decomposition of $\Gamma_0$ using the invasion graph. The proof closely follows the logic of \citet{hofbauer-schreiber-2022-permanence} but is included for the convenience of the reader.

\begin{lemma} \label{lemma:acyclic:build:morse:demposition}
Let $\ell \in \{ 0,1,\dots , k-1\}$ and $\C_\ell =\{I \in \C: \abs{I} \leq \ell\}$. Suppose $\IG$ is acyclic. Then, for each $I \in \C_\ell$ there is a nonempty compact invariant subset $M_I  \subset \S_I$ such that
\begin{itemize}
    \item[1.] $\omega(z)\subset \bigcup_{I \in \C_\ell} M_I$ for all $z\in \bigcup_{I \in \C_\ell} \S_I$.
    \item[2.] For each negative orbit $\mathcal{O}^- \subset \Gamma_0 \cap \bigcup_{I \in \C_\ell} \S_I$, $\alpha(\mathcal{O}^-)\subset \bigcup_{I \in \C_\ell} M_I$.
    \item[3.]  Each $M_I$ is isolated in $\S$.
     \item[4.] The family of invariant sets $\{ M_I: I \in \C_\ell \}$ is a Morse decomposition of $\bigcup_{I \in \C_\ell} \S_I \cap \Gamma_0$.
\end{itemize}
\end{lemma}


\begin{proof}[Proof of Lemma~\ref{lemma:acyclic:build:morse:demposition}.] We prove this lemma by induction on $\ell$. If $\ell = 0$, then $M_{\emptyset} = \{z=(0,y):z\in \Gamma_0\}$ is hyperbolic by Assumptions \textbf{A6}-\textbf{A7}, and, consequently,  isolated in $\S_0$. Properties 1, 2, and 4, hold immediately. 

Suppose that the lemma holds for $\ell-1$, and all the $M_I$ for $\abs{I} < \ell$ are given.  Let $I\in \C_\ell$. By induction (property 4), the family of $M_J$ such that $J\subsetneq I$ and $J\in \C_{\ell-1}$ form a Morse decomposition of the boundary $\partial (\S_I\cap \Gamma)$ of $\S_I\cap \Gamma$. By induction (property 3), each $M_J$ is isolated in $\S$. Hence, $\partial (\S_I\cap \Gamma)$ is isolated in $\S$ and there exists a maximal compact invariant set $M_I$ in $\S_I$. 

To show property 3., suppose to the contrary that $M_I$  is not isolated in $\S$. Lemma~\ref{lem:key:two} implies there exist invariant measures $\mu_+,\mu_-$ such that $\mu_+(M_I)=\mu_-(M_I)=1$ and $r_j( \mu_+) \ge 0$ and $r_j( \mu_-) \le 0$ for all $j\in \{1,2,\dots,k\}$. Assumption \textbf{A7} implies that $r_j( \mu_+) > 0$ and $r_j( \mu_-) < 0$ for $j\notin I$. This contradicts assumption \textbf{A6}. Thus, property 3. holds. 

Finally, the assumption that $\IG$ is acyclic implies property 4 by choosing a suitable order on $\S_\ell$.
\end{proof}

Taking $\ell =k-1$ in Lemma~\ref{lemma:acyclic:build:morse:demposition}, one gets that the family of invariant sets $\{ M_I: I \in \S \}$ is a Morse decomposition of $\Gamma_0$. Applying Theorem~\ref{thm:main} completes the proof.

\end{document}